\documentclass[lettersize,journal]{IEEEtran}
\usepackage{amsmath,amsfonts,amssymb}
\usepackage{algorithmic}
\usepackage{algorithm}
\usepackage{tabularx}
\usepackage{array}
\usepackage{breqn}
\usepackage[caption=false,font=normalsize,labelfont=sf,textfont=sf]{subfig}
\usepackage{textcomp}
\usepackage{stfloats}
\usepackage{url}
\usepackage{verbatim}
\usepackage{graphicx}
\usepackage{cite}
\usepackage{amsthm}
\usepackage{booktabs}
\usepackage{arydshln}
\usepackage{pifont}
\usepackage{balance}
\usepackage{multirow}
\newcommand{\cmark}{\ding{51}}
\newtheorem{thm}{Theorem}

\begin{document}

\title{Stacked Intelligent Metasurfaces for Near-Field Multi-User Covert Communications}

\author{Ahmed M. Benaya, \IEEEmembership{Member, IEEE}, Ali A. Nasir, \IEEEmembership{Senior Member, IEEE}, Khaled M. Rabie, \IEEEmembership{Senior Member, IEEE}, and Daniel B. da Costa, \IEEEmembership{Senior Member, IEEE}
\thanks{This work is supported by the Deanship of Research at King Fahd University of Petroleum and Minerals (KFUPM) for funding under the Interdisciplinary Research Center for Communication Systems and Sensing (IRC-CSS).}
\thanks{Ahmed M. Benaya is with the Interdisciplinary Research Center for Communication Systems and Sensing (IRC-CSS), King Fahd University of Petroleum and Minerals (KFUPM), Dhahran 31261, Saudi Arabia. (e-mail: ahmed.diab@kfupm.edu.sa)}
\thanks{Ali A. Nasir and Daniel B. da Costa are with the Department of Electrical Engineering, King Fahd University of Petroleum and Minerals (KFUPM),  Dhahran 31261, Saudi Arabia, while Khaled M. Rabie is with the Department of Computer Engineering, King Fahd University of Petroleum and Minerals (KFUPM),  Dhahran 31261, Saudi Arabia. They all are also affiliated with the Center for Communication Systems and Sensing at KFUPM. (e-mail: anasir@kfupm.edu.sa, k.rabie@kfupm.edu.sa, danielbcosta@ieee.org).}}



\maketitle

\begin{abstract}
Reconfigurable intelligent surfaces have emerged as a cutting-edge technology for next-generation wireless communications that are capable of reconfiguring the wireless environment using a large number of cost-effective reflecting elements. However, a significant body of prior studies has focused on single-layer surfaces that lack the capability of significantly mitigating inter-user interference. Moreover, previous studies mostly consider far-field operation and neglect working in the near-field region. In this paper, we propose a stacked intelligent metasurfaces (SIM)-assisted near-field multi-user multiple-input-single-output covert communication system. More specifically, we have a multi-antenna base station that is assisted with a SIM to serve multiple single-antenna users in the presence of multiple single-antenna wardens. We aim at optimizing the beamfocusing vectors at the BS and SIM phase shift matrices to maximize the sum covert rate under maximum transmit power budget constraint, quality-of-service (QoS) constraint for all users, and covertness constraint. Since the formulated problem is highly non-convex due to the coupling between the variables, we adopt alternating optimization to tackle it, where we divide the problem into beamfocusing sub-problem and SIM phase shift sub-problem, which are solved alternately until convergence. We leverage successive convex approximation (SCA) to solve the two sub-problems. Additionally, we formulate the SIM phase shift sub-problem using the widely adopted projected gradient ascent (PGA) method for comparison purposes. The conducted simulations reveal that the SCA-based algorithm outperforms the existing PGA-based algorithm as well as other benchmarks in terms of the achieved sum covert rate, demonstrating its consistent performance and robustness under various system parameter configurations.
\end{abstract}


\section{Introduction}
\IEEEPARstart{T}{he} global vision for sixth-generation (6G) wireless systems is defined by a demand for unprecedented data rates, ultra-low latency, and ubiquitous connectivity~\cite{Wang2023On}. However, the performance of 6G wireless systems is fundamentally constrained by the wireless propagation environment, which is often uncontrollable and suffers from blockages and multi-path fading~\cite{liaskos2018new}. Reconfigurable intelligent surface (RIS) technology has emerged as a paradigm shifting solution for next-generation wireless networks due to its ability to control the wireless environment and overcome the limitations of conventional wireless communication systems~\cite{liu2021reconfigurable,mei2022intelligent,cao2023physical,hua2023secure,benaya2024double}. RISs are intelligent surfaces consisting of a large number of reconfigurable passive/active elements that can dynamically adjust their reflection properties to steer wireless signals toward target receivers, thereby mitigating blockages, overcoming non-line-of-sight (NLOS) conditions, and enhancing link quality~\cite{wu2021intelligent,salem2022intelligent}.

A large portion of the existing research efforts on RIS-aided wireless systems has focused on single-layer metasurface architectures. Although such designs are simple in implementation, they inherently restrict the available degrees-of-freedom for reconfiguring the wireless environment, thereby limiting their capability to effectively suppress inter-user interference in multi-user communication scenarios~\cite{wei2022multi,guo2020weighted}. To overcome these inherent limitations, stacked intelligent metasurfaces (SIMs) have recently emerged as the latest development in intelligent surface technology, providing energy-efficient signal processing in the wave domain~\cite{an2025stacked,li2024stacked}. Unlike conventional single-layer RIS, SIMs are composed of multiple cascaded metasurface layers, significantly improving the degrees-of-freedom available for wave manipulation~\cite{liu2025stacked}. This multi-layer architecture enables more flexible and fine-grained control of the electromagnetic propagation environment, thereby improving interference management and overall system performance in complex multi-user wireless scenarios.

In addition to advanced intelligent surface architectures, operating in high-frequency bands above 6 GHz offers significantly higher data rates for future wireless systems. The integration of high-frequency operation with large intelligent surfaces fundamentally reshapes the propagation characteristics of 6G networks by substantially expanding the near-field region. This can be attributed to the increased Rayleigh distance, which scales with both the carrier frequency and the physical aperture size of the surface~\cite{cui2022near}. Within the near-field regime, electromagnetic wave propagation is modeled by spherical waves rather than conventional planar waves for far-field scenarios, providing enhanced beamfocusing capabilities~\cite{liu2024near}. Although the adoption of large-scale SIMs can effectively improve coverage and spectral/energy efficiencies at high frequencies by creating additional controllable propagation paths, they also unintentionally increase the vulnerability of transmitted signals to be intercepted by unintended receivers~\cite{benaya2023physical}. Consequently, the development of SIM-assisted near-field covert transmission strategies is crucial to minimize signal detectability while maintaining reliable communication.

\subsection{Related Work}
Recently, significant research efforts have been devoted to investigating SIM-assisted wireless communication systems~\cite{an2025stacked,liu2025multi,bahingayi2025refined,niu2024stacked,ebrahimi2025stacked}. In~\cite{an2025stacked}, a SIM-assisted downlink multi-user multiple-input-single-output (MISO) system has been proposed to eliminate the need for digital beamforming by optimizing the transmit power at the base station (BS) and the phase shift of the SIM such that the sum rate is maximized. In that work, the authors proposed a computationally-efficient algorithm to solve the formulated problem. A customized deep reinforcement learning (DRL)-based approach has been proposed in~\cite{liu2025multi} to tackle the transmit power and SIM phase shift optimization in a similar multi-user MISO downlink scenario. Additionally, the same problem has been investigated with the design of a digital beamformer at BS along with the SIM phase shifts in~\cite{bahingayi2025refined}. The authors in that work concluded that the order of alternating optimization (AO) greatly affects the achieved performance. Moreover, iterative optimization of the SIM phase shift sub-problem may further enhance the system performance. Integration of SIM-assisted networks with integrated sensing and communication (ISAC) has been considered in~\cite{niu2024stacked}, where a SIM-enabled base station generates a beam pattern for simultaneously communicating with multiple downlink users and detecting a target. The authors formulated an optimization problem to maximize spectral efficiency through optimizing the SIM phase shifts and the power allocation at the base station. Also, in~\cite{ebrahimi2025stacked}, both SIM and simultaneously transmitting and reflecting RIS have been integrated together in THz ISAC system to address the challenge of multi-user communication and multi-target sensing. In that work, the SIM has been used to enable high-precision wavedomain beamforming, while the STAR-RIS has been used to achieve full-space coverage for improved multi-user sensing and communication capabilities.

Fewer research attempts have focused on SIM-assisted systems operating in the near-field regime~\cite{jia2024stacked,papazafeiropoulos2024near,li2025stacked}. In~\cite{jia2024stacked}, a SIM-assisted multi-user MISO near-field system has been proposed, where the SIM-enabled BS perform beamfocusing in the wave domain instead of the conventional digital beamfocusing. A SIM-assisted multi-user multiple-input-multiple-output (MIMO) system that operates in the near-field region has been proposed in~\cite{papazafeiropoulos2024near}. The authors in that work formulated a weighted sum rate maximization problem that optimizes the power allocation at the BS as well as the SIM phase shifts, where block coordinate descent algorithm is utilized to solve the formulated problem. In~\cite{li2025stacked}, a hybrid beamforming framework for a SIM-assisted wideband holographic multi-user MIMO system adopting a near-field channel model has been proposed, where the holographic beamformer is designed to maximize aggregate eigen-channel gains. More specifically, a minimum mean square error (MMSE)-based digital tramsnit precoding and iterative water-filling are employed to enable efficient multi-user transmission and power allocation. In addition, a layer-by-layer iterative algorithm is developed to optimize SIM phase shifts.

In the context of secure communication, SIMs have been adopted to enhance the performance of single-input-single-output (SISO) wireless systems by exploiting the additional spatial degrees of freedom introduced by multi-layer metasurface architectures. In particular, the authors in~\cite{niu2024enhancing,niu2024efficient} proposed a SIM-assisted SISO transmission framework in which the SIM is deployed at the transmitter to directly manipulate electromagnetic waves, enabling joint modulation, beamforming, and artificial noise generation in the wave domain. Closed-form solutions and AO algorithms were developed to configure SIM phase shifts and transmit power efficiently, aiming at improving secrecy rate without relying on multiple radio frequency (RF) chains. In that work, simulation results demonstrate that SIMs can effectively compensate for the limited spatial degrees of freedom in conventional SISO systems and significantly outperform benchmark single-layer or antenna-based schemes in terms of achievable secrecy performance.
\subsection{Motivation and Contribution}
A closer examination of the aforementioned studies reveals that most of the research efforts on SIM-assisted systems have focused on far-field operation and neglected the expanded near-field region because of the large SIM aperture size and potential high frequencies. Moreover, the few research attempts that considered near-field high-frequency SIM-assisted systems have neglected the vulnerability of the network to security issues due to the inherent broadcast nature of the channel. Although~\cite{niu2024enhancing,niu2024efficient} have investigated the effectiveness of SIM-assisted systems in improving physical-layer security, their designs generally rely on secrecy rate maximization, which assumes that the transmission is detectable by potential eavesdroppers. However, in emerging near-field high-frequency systems, preventing signal detection itself becomes a more fundamental requirement. Motivated by this fact, in this work we propose the adoption of covert communication frameworks that aims to minimize transmission detectability rather than merely securing the transmitted information. To the best of our knowledge, this is the first work to consider a SIM-assisted near-field multi-user covert communication system. The novel contributions of our work compared to the state-of-the-art are summarized in Table~\ref{tab:comparison}. The main contributions are as follows;
\begin{itemize}
    \item We propose a SIM-assisted multi-user MISO downlink covert communication framework operating in the near-field region. The proposed framework aims at maximizing the sum covert rate at all user equipment (UEs), while maintaining a minimum quality-of-service (QoS) at all UEs and minimizing the detection capabilities at multiple potential wardens.
    \item We analyze the detection error probability at each warden and derive the covertness constraint that ensures minimizing the detection capabilities at wardens.
    \item We formulate a sum covert rate maximization problem in which we design the digital beamfocusing vectors at the BS as well as the SIM phase shift matrices under minimum QoS at all UEs, power budget, phase shift, and covertness constraints. Due to the coupling between variables, we adopt AO to tackle the formulated problem, where we propose successive convex approximation (SCA)-based algorithms to solve both beamfocusing and phase shift problems alternately. Additionally, we adopt the widely used projected gradient ascent (PGA) algorithm to solve the phase shift problem and compare its performance with the SCA-based scheme.
    \item Extensive simulation results are provided to illustrate how different system parameters affect the achievable sum covert rate and to demonstrate the superiority of the proposed scheme over existing benchmarks.
\end{itemize}

\begin{table*}[t]
\centering
\caption{Comparison of Our Work with Existing SIM-Assisted Wireless Systems}
\label{tab:comparison}
\setlength{\tabcolsep}{6pt}
\renewcommand{\arraystretch}{1.2}
\begin{tabular}{l ccccccccccccc}
\toprule
 & \textbf{Our paper} & \cite{an2025stacked} & \cite{liu2025multi} & \cite{bahingayi2025refined} & \cite{niu2024stacked} & \cite{ebrahimi2025stacked} & \cite{jia2024stacked} & \cite{papazafeiropoulos2024near} & \cite{li2025stacked} \\
\midrule
Multi-user access                 & \cmark & \cmark & \cmark & \cmark & \cmark & \cmark & \cmark & \cmark & \cmark \\
\hdashline
Near-field channel model          & \cmark &  &  &  &  &  & \cmark &  \cmark& \cmark \\
\hdashline
Digital beamfocusing design                & \cmark &  &  & \cmark & &  &  &  & \cmark \\
\hdashline
PGA-based phase shift design          & \cmark & \cmark &  & \cmark & \cmark &  & \cmark & \cmark & \\
\hdashline
SCA-based phase shift design          & \cmark &  &  &  &  &  &  &  & \\
\hdashline
Covert communication & \cmark &  &  &  &  &  &  &  &  \\
\bottomrule
\end{tabular}
\end{table*}

The remainder of the paper is organized as follows; The SIM-assisted multi-user system and near-field channel models are presented in Section~\ref{sec_sys}. Section~\ref{sec_covert} provides the error detection probability analysis for different wardens. In Section~\ref{sec_opt}, the optimization problem formulation and the corresponding solution are presented, with some detailed derivations given in the Appendix. Simulation results and discussions are presented in Section~\ref{results}. Finally, Section~\ref{sec_conc} concludes the paper and outlines several directions for future research.

\textit{Notations}: Matrices and vectors are denoted by bold uppercase and lowercase letters, respectively. For a vector $\mathbf{a}$, $\mathrm{diag}(\mathbf{a})$ denotes a diagonal matrix whose main diagonal entries are given by the elements of $\mathbf{a}$. The operators $(\cdot)^T$ and $(\cdot)^H$ represent the transpose and Hermitian (conjugate transpose), respectively. The symbols $|\cdot|$, $\|\cdot\|$, and $\mathbb{E}[\cdot]$ denote the absolute value, the Euclidean norm, and the expectation operator. Moreover, $\mathcal{R}[\cdot]$ and $\mathcal{I}[\cdot]$ correspond to the real and imaginary parts of a complex-valued scalar, respectively. The operators $\lfloor \cdot \rfloor$, $\lceil \cdot \rceil$, and $\mathrm{mod}(\cdot)$ denote the floor, ceiling, and modulo operations, respectively.
\section{System Model}
\label{sec_sys}
In this section, we describe the proposed SIM-assisted near-field covert communication system, the signaling model, and the adopted near-field channel model.
\subsection{System Description and Signaling model}
\begin{figure}[!t]
\centerline{\includegraphics[width=0.5\textwidth]{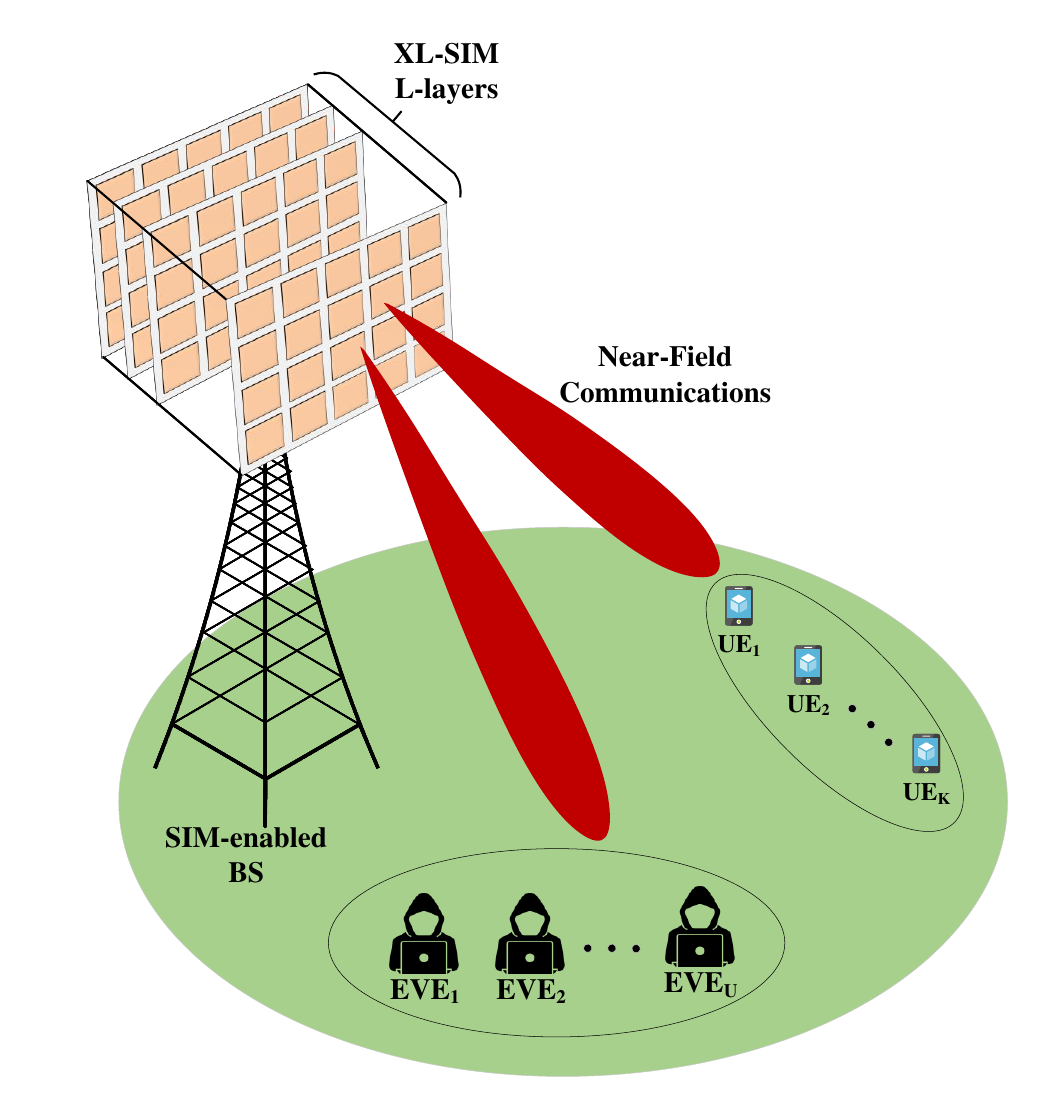}}
\caption{SIM-assisted multi-user multi-warden near-field downlink covert communication system.}
\label{fig_1}
\end{figure}
As depicted in Fig.~\ref{fig_1}, we consider an SIM-assisted downlink near-field covert communication system, which consists of a BS equipped with a uniform linear array (ULA) of $M$-antennas, a group of $K$ single-antenna legitimate UEs, and a group of $U$ single-antenna wardens. The BS operates in the high-frequency band and is assisted by a SIM that comprises $L$ equidistance-spaced reconfigurable metasurface layers, each composed of $N$ elements. Let $\mathcal{K} = \left\{1,2,\dots,K\right\}$, $\mathcal{U} = \left\{1,2,\dots,U \right\}$, $\mathcal{L} = \left\{1,2,\dots,L \right\}$, and $\mathcal{N} = \left\{1,2,\dots,N \right\}$ represent the sets of legitimate UEs, wardens, metasurface layers, and metasurface elements in each layer, respectively. We denote the complex-valued coefficient of the $n$-th element in the $l$-th layer by $\alpha_n^{l} e^{j\theta_n^l},~\forall n\in \mathcal{N}, l\in\mathcal{L}$, where $\alpha_n^{l} \in \left[0,1\right]$ and $\theta_n^l \in \left[ 0,2\pi\right)$ represent the amplitude and phase shift of the $n$-th element in the $l$-th layer, respectively. For maximum energy transfer, we set $\alpha_n^{l} = 1,~\forall n\in \mathcal{N}, l\in\mathcal{L}$. Hence, the diagonal phase shift matrix $\mathbf{\Theta}^l$ of the $l^{th}$ metasurface layer is given by:
\begin{equation}
    \mathbf{\Theta}^l = \text{diag} \left( e^{j\theta_1^l}, e^{j\theta_2^l}, \dots e^{j\theta_N^l}\right)\in\mathbb{C}^{N\times N}, ~\forall l\in \mathcal{L}. \label{Eq:1}
\end{equation}

Each SIM layer is modeled as a rectangular uniform planner array (UPA) of $N= N_x \times N_z$ elements, where $N_x$ and $N_z$ represent the number of elements in the $x$-axis and $z$-axis directions, respectively. Following the Rayleigh–Sommerfeld diffraction model~\cite{lin2018all,liu2022programmable}, the inter-layer propagation coefficient from the $\Tilde{n}$-th element in the $(l-1)$-th SIM layer to the $n$-th element in the $l$-th SIM layer, $\forall n,\Tilde{n} \in \mathcal{N}, l \in \mathcal{L}$, can be written as:
\begin{equation}
    w_{n,\Tilde{n}}^l = \frac{A \cos{\chi^l_{n,\Tilde{n}}}}{r^l_{n,\Tilde{n}}} \left( \frac{1}{2 \pi r^l_{n,\Tilde{n}}} - j\frac{1}{\lambda}\right) e^{2 \pi r^l_{n,\Tilde{n}}/\lambda}, \label{Eq:2}
\end{equation}
where $A$, $\chi^l_{n,\Tilde{n}}$, $r^l_{n,\Tilde{n}}$, and $\lambda$ are the area of each element, the angle between the propagation direction and the normal direction of the $\left(l-1 \right)$-th SIM layer, the propagation distance, and the wavelength, respectively. 

Hence, the overall SIM response can be expressed as:
\begin{equation}
    \mathbf{G} = \mathbf{\Theta}^L \mathbf{W}^L \dots \mathbf{\Theta}^2 \mathbf{W}^2 \mathbf{\Theta}^1 \mathbf{W}^1 \in \mathbb{C}^{N\times M}, \label{Eq:3}
\end{equation}
where $\mathbf{W}^l \in \mathbb{C}^{N\times N},~\forall~l\in\mathcal{L}/\left\{ 1\right\}$ is the propagation coefficient matrix from the $\left(l-1 \right)$-th SIM layer to the $l$-th SIM layer. The propagation coefficient matrix from the $1$-st SIM layer to the transmit antenna layer, $\mathbf{W}^1 \in \mathbb{C}^{N\times M}$, can also be obtained from~\eqref{Eq:2} by replacing $r^l_{n,\Tilde{n}},\chi^l_{n,\Tilde{n}},~\forall n,\Tilde{n} \in \mathcal{N}$ with $r^1_{n,m},\chi^1_{n,m}~\forall n \in \mathcal{N}, m \in \mathcal{M}$.

Let $\mathbf{x} = [x_1, x_2, \dots, x_K]^T \in \mathbb{C}^{K \times 1}$ denotes the column vector of the transmitted symbols for all UEs, where $\mathbb{E}[\left\| \mathbf{x} \right\|^2] = 1$ and $x_k$ represents the transmitted symbol to the $k$-th UE. Hence, the received signal at the $k$-th UE can be written as:
\begin{align}
    y_k &= \mathbf{h}_k^H \mathbf{G} \mathbf{V} \mathbf{x} + n_k \nonumber \\ &= \mathbf{h}_k^H \mathbf{G} \mathbf{v}_k x_k + \sum_{i = 1, i\neq k}^K \mathbf{h}_k^H \mathbf{G} \mathbf{v}_i x_i + n_k, \label{Eq:4}
\end{align}
where $\mathbf{h}_k \in \mathbb{C}^{N \times 1}$ is the near-field channel vector from the $L$-th SIM layer to the $k$-th UE, $\mathbf{V} \in \mathbb{C}^{M \times K}$ is the digital beamfocusing matrix at the BS with $\mathbf{v}_k \in \mathbb{C}^{M \times 1}$ denoting the beamfocusing vector of the $k$-th UE, and $n_k$ is the additive white Gaussian noise (AWGN) component with variance $\sigma_n^2$, i.e. $n_k \sim \mathcal{CN}\left( 0,\sigma_n^2\right)$. The received signal-to-interference-plus-noise ratio (SINR) at the $k$-th UE can be expressed as:
\begin{equation}
    \gamma_k = \frac{\left| \mathbf{h}_k^H \mathbf{G} \mathbf{v}_k\right|^2}{\sum_{i = 1, i\neq k}^K \left| \mathbf{h}_k^H \mathbf{G} \mathbf{v}_i\right|^2 + \sigma_n^2}. \label{Eq:5}
\end{equation}

The achieved sum rate of all UEs can be written as:
\begin{equation}
    R = \sum_{k = 1}^K R_k = \sum_{k = 1}^K \log_2 \left( 1+ \gamma_k \right). \label{Eq:6}
\end{equation}

\subsection{Near-Field Channel Model}
In this work, we adopt the uniform surface wave (USW) near-field line-of-sight (LoS) channel model~\cite{liu2023near}. As mentioned earlier, the SIM is placed in the $x$-$z$ plane and is composed of $N=N_x\times N_z$ elements, where $N_x = 2 \Tilde{N}_x +1$ is the number of elements in the $x$-direction and $N_z = 2 \Tilde{N}_z +1$ is the number of elements in the $z$-direction. Let the inter-element distance in the $x$-direction and $z$-direction be $d_x$ and $d_z$, respectively. Hence, the near-field LoS channel from the $L$-th layer to the $k$-th UE can be modeled as:
\begin{equation}
    \mathbf{h}_k = \beta_k e^{-j \frac{2\pi}{\lambda}r_k} \left[ \mathbf{a}_k^x \left( \theta_k, \phi_k, r_k\right) \otimes \mathbf{a}_k^z \left( \phi_k, r_k\right)\right], \label{Eq:7}
\end{equation}
where $\mathbf{a}_k^x \left( \theta_k, \phi_k, r_k\right)$ and $\mathbf{a}_k^z \left( \phi_k, r_k\right)$ are column vectors representing the array response in the $x$- and $z$-directions, respectively. In addition, $\beta_k = \beta_0 r_k^{-\eta}$, $r_k$, $\theta_k$, and $\phi_k$ denote the large-scale fading coefficient, the propagation distance, the azimuth angle, and the elevation angle of the link between the $L$-th SIM layer and the $k$-th UE, respectively, where $\beta_0 = \left(\frac{\lambda}{4\pi}\right)^2$ is the free space path loss at a reference distance $1$ m and $\eta$ is the path loss exponent. Assuming $d_x/r_k \ll 1$ and $d_z/r_k \ll 1$, and exploiting the Fresnel approximation~\cite{liu2023near}, the elements of the array response can be written as:
\begin{dmath}
    \left[ \mathbf{a}_k^x \left( \theta_k, \phi_k, r_k\right)\right]_{n_x} = e^{-j \frac{2 \pi}{\lambda} \left( - n_x d_x \cos{\theta_k} \sin{\phi_k} + \frac{n_x^2 d_x^2 \left( 1 - \cos^2{\theta_k} \sin^2{\phi_k}\right)}{2r_k}\right)}, \label{Eq:8}
\end{dmath}
and
\begin{equation}
    \left[ \mathbf{a}_k^z \left(\phi_k, r_k\right)\right]_{n_z} = e^{-j \frac{2 \pi}{\lambda} \left( - n_z d_z \cos{\phi_k} + \frac{n_z^2 d_z^2 \sin^2{\phi_k}}{2r_k}\right)}, \label{Eq:9}
\end{equation}
where $\left[ \mathbf{a}_k^x \left( \theta_k, \phi_k, r_k\right)\right]_{n_x}$ and $\left[ \mathbf{a}_k^z \left(\phi_k, r_k\right)\right]_{n_z}$ represent the $n_x$-th element in the column vector $\mathbf{a}_k^x \left( \theta_k, \phi_k, r_k\right)$ and the $n_z$-th element in the column vector $ \mathbf{a}_k^z \left(\phi_k, r_k\right)$, respectively, where $n_x \in \left\{ -\Tilde{N}_x, \dots, \Tilde{N}_x\right\}$ and $n_z \in \left\{ -\Tilde{N}_z, \dots, \Tilde{N}_z\right\}$. Similarly, the near-field LoS channel from the $L$-th SIM layer to the $u$-th warden, $\forall u\in \mathcal{U}$, can be modeled as:
\begin{equation}
    \mathbf{h}_u = \beta_u e^{-j \frac{2\pi}{\lambda}r_u} \left[ \mathbf{a}_u^x \left( \theta_u, \phi_u, r_u\right) \otimes \mathbf{a}_u^z \left( \phi_u, r_u\right)\right], \label{Eq:7b}
\end{equation}
where the array response in the $x$- and $z$-directions, $\mathbf{a}_u^x \left( \theta_u, \phi_u, r_u\right)$ and $\mathbf{a}_u^z \left( \phi_u, r_u\right)$, can be defined in similar ways to~\eqref{Eq:8} and~\eqref{Eq:9} . Additionally, $\beta_u=\beta_0r_u^{-\eta}$, $r_u$, $\theta_u$, and $\phi_u$ denote the large-scale fading coefficient, the propagation distance, the azimuth angle, and the elevation angle of the link between the $L$-th SIM layer and the $u$-th warden, respectively. In this work, we adopt a practical insider-threat model in which the wardens are initially legitimate users that are later identified as suspicious based on abnormal behavior detected via upper-layer network protocols. Consequently, their locations and CSI are assumed to be known at the BS through prior channel estimation and signaling procedures~\cite{liu2024ris, hu2025reconfigurable, chen2021enhancing}.
\section{Detection Performance Analysis}
\label{sec_covert}
Following~\cite{zhou2025star,liu2024ris}, the detection error probability (DEP) is utilized as the covertness metric in the proposed covert communication system. With common equal prior probabilities, let $\xi_u = P_{\text{FA}_u} + P_{\text{MD}_u}$ denotes the DEP at the $u$-th warden, where $P_{\text{FA}_u}$ and $P_{\text{MD}_u}$ represent the probabilities of false alarm and missed detection at the $u$-th warden, respectively. Let us assume that all wardens make $J$ observations in each time slot as in~\cite{jiang2021resource}. Hence, the received signal in the $j$-th observation at the $u$-th warden can be expressed as:
\begin{equation}
    y_u^{\left(j\right)} = \left\{\begin{matrix}
n_u^{\left(j\right)}, &  \mathcal{H}_0,\\
\sum_{i = 1}^K \mathbf{h}_u^H \mathbf{G} \mathbf{v}_i x_i^{\left(j\right)} + n_u^{\left(j\right)}, & \mathcal{H}_1, \\
\end{matrix}\right. \label{Eq:10}
\end{equation}
where $\mathcal{H}_0$ denotes the null hypothesis that the BS is not transmitting and $\mathcal{H}_1$ denotes the alternative hypothesis that the BS is transmitting. In addition, $n_u^{\left(j\right)} \sim \mathcal{CN}\left( 0,\sigma_n^2\right)$ is the AWGN component at the $u$-th warden. Let the received vector at the $u$-th warden be expressed as:
\begin{equation}
    \mathbf{y}_u = \left[ y_u^{\left(1\right)}, y_u^{\left(2\right)}, \dots, y_u^{\left(J\right)}\right]^T. \label{Eq:11}
\end{equation}

Since all the elements of $\mathbf{y}_u$ are independent and identical complex Gaussian random variables, the likelihood functions of $\mathbf{y}_u$ under both $\mathcal{H}_1$ and $\mathcal{H}_0$ can be expressed as follows:
\begin{equation}
    \left\{\begin{matrix}
\mathbb{P}_{1,u} = p_{\mathcal{H}_1,u}\left ( \mathbf{y}_u \right ) = \frac{1}{\pi^J \delta_{1,u}^J}e^{\frac{-\mathbf{y}_u^H \mathbf{y}_u}{\delta_{1,u}}} \\ \mathbb{P}_{0,u} = p_{\mathcal{H}_0,u}\left ( \mathbf{y}_u \right ) = \frac{1}{\pi^J \delta_{0,u}^J}e^{\frac{-\mathbf{y}_u^H \mathbf{y}_u}{\delta_{0,u}}},
\end{matrix}\right. \label{Eq:12}
\end{equation}
where $\delta_{1,u} = \sum_{i = 1}^K |\mathbf{h}_u^H \mathbf{G} \mathbf{v}_i|^2 + \sigma_u^2$ and $\delta_{0,u} = \sigma_u^2$. Assuming the minimum DEP at the $u$-th warden is $\xi_u^{\ast}$, the covertness constraint can be set as $\xi_u^{\ast} \ge 1-\epsilon$, where $\epsilon \in [0,1]$ is a very small positive number to ensure covertness. According to~\cite{yan2019gaussian}, for the optimal detector, the minimum DEP at the $u$-th warden is $\xi_u^{\ast} = 1-\mathcal{V}_u \left(\mathbb{P}_{0,u},\mathbb{P}_{1,u}\right)$, where $\mathcal{V}_u \left(\mathbb{P}_{0,u},\mathbb{P}_{1,u}\right)$ is the total variation distance (TVD) between $\mathbb{P}_{0,u}$ and $\mathbb{P}_{1,u}$. Following~\cite{yan2019gaussian,zhang2024robust}, the TVD between $\mathbb{P}_{0,u}$ and $\mathbb{P}_{1,u}$ can be expressed in terms of Pinsker's inequality as follows: 
\begin{equation}
    \mathcal{V}_u \left(\mathbb{P}_{0,u},\mathbb{P}_{1,u}\right) \le \sqrt{\frac{\mathcal{D}_u\left(\mathbb{P}_{1,u}\|\mathbb{P}_{0,u}\right)}{2}}, \label{Eq:13}
\end{equation}
where $\mathcal{D}_u\left(\mathbb{P}_{1,u}\|\mathbb{P}_{0,u}\right)$ is Kullback-Leibler (KL) divergence, which can be expressed as:
\begin{align}
\mathcal{D}_u\left(\mathbb{P}_{1,u}\|\mathbb{P}_{0,u}\right) &= \int_{-\infty}^{\infty} \mathbb{P}_{1,u} \ln{\left(\frac{\mathbb{P}_{1,u}}{\mathbb{P}_{0,u}} \right)} d\mathbf{y}_u \nonumber \\ &=J\nu\left(\frac{\delta_{1,u}-\delta_{0,u}}{\delta_{0,u}}\right),\label{Eq:14}
\end{align}
where $\nu\left(x\right) = x-\ln\left(1+x\right)$. Consequently, we can write $\xi_u^{\ast}  \ge 1-\sqrt{\frac{\mathcal{D}_u\left(\mathbb{P}_{1,u}\|\mathbb{P}_{0,u}\right)}{2}}$. Hence, the covertness constraint can be set as:
\begin{equation}
\mathcal{D}_u\left(\mathbb{P}_{1,u}\|\mathbb{P}_{0,u}\right) \le 2\epsilon^2. \label{Eq:15}
\end{equation}


\section{Optimization Problem Formulation and Solution}
\label{sec_opt}
The main aim of this work is to design the beamfocusing vectors at the BS and the SIM phase shift matrices to maximize the sum covert rate of all UEs under minimum rate requirements, power, and covert communication constraints. This can be formulated mathematically as follows:
\begin{subequations}
\begin{align}
	\mathcal{P}1: \max_{\left\{\mathbf{v}_k\right\}_{k=1}^{K}, \left\{\mathbf{\Theta}^l \right\}_{l=1}^{L}} \sum_{k=1}^{K} \log_2 \left(1+\gamma_k\right), \label{Eq:16a} \\
	& \hspace{-2in} \textrm{Subject to:} \nonumber \\
	& \hspace{-1.5in} R_k \ge R_k^{\text{min}},~\forall~k\in \mathcal{K}, \label{Eq:16b}\\
 & \hspace{-1.5in} \sum_{k=1}^K \left\| \mathbf{v}_k\right\|^2 \le P_{\text{max}},~\forall~k\in \mathcal{K}, \label{Eq:16c} \\
 & \hspace{-1.5in} \left| \theta_n^l\right| = 1,~\forall~n \in \mathcal{N}, l\in \mathcal{L}, \label{Eq:16d} \\
 & \hspace{-1.5in} \mathcal{D}_u\left(\mathbb{P}_{1,u}\|\mathbb{P}_{0,u}\right) \le 2\epsilon^2,~\forall~ u \in \mathcal{U}. \label{Eq:16e}
\end{align}
\label{Eq:16}
\end{subequations}

Constraint~\eqref{Eq:16b} ensures that the achieved covert rate at each UE is above the minimum covert rate
requirement $R_k^{\text{min}}$, while constraint~\eqref{Eq:16c} ensures that the total transmit power at the BS is below the maximum power budget $P_{\text{max}}$. Additionally, constraint~\eqref{Eq:16d} ensures that the amplitude and the phase shift of each SIM meta atom is set as $\alpha_n^{l} = 1$ and $\theta_n^l \in \left[ 0,2\pi\right)~\forall n\in \mathcal{N}, l\in\mathcal{L}$. Finally, constraint~\eqref{Eq:16e} ensures that optimal DEP at each warden is above a given threshold to achieve covertness. It's clear that problem $\mathcal{P}1$ is highly non-convex with respect to the optimization variables $\left\{\mathbf{v}_k\right\}_{k=1}^{K}$ and $\left\{\mathbf{\Theta}^l \right\}_{l=1}^{L}$ due to the coupling between the variables in the non-convex objective function,~\eqref{Eq:16a}, and the non-convex  constraints,~\eqref{Eq:16b} and~\eqref{Eq:16d}. To tackle this issue, we are going to decouple $\mathcal{P}1$ into two disjoint sub-problems and adopt AO until convergence. In the first sub-problem we will design the beamfocusing vectors for fixed phase shift matrices at different SIM layers, while in the second sub-problem, we will design the phase shift matrices for fixed beamfocusing vectors.
\subsection{Beamfocusing Vectors Optimization}
Assuming fixed phase shift matrices $\left\{\mathbf{\Theta}^l \right\}_{l=1}^{L}$, the beamfocuing vectors optimization problem can be formulated as:
\begin{subequations}
\begin{align}
	\mathcal{P}2: \max_{\left\{\mathbf{v}_k\right\}_{k=1}^{K}} \sum_{k=1}^{K} \log_2 \left(1+\gamma_k\right), \label{Eq:17a} \\
	& \hspace{-1.5in} \textrm{Subject to:} \nonumber \\
	& \hspace{-1in} \eqref{Eq:16b},~\eqref{Eq:16c},~\text{and}~\eqref{Eq:16e}. \label{Eq:17b}
\end{align}
\label{Eq:17}
\end{subequations}

Problem $\mathcal{P}2$ remains non-convex due to the fractional nature of the SINR term $\gamma_k$ in the objective function~\eqref{Eq:17a} and in the rate requirement constraint~\eqref{Eq:16b}. To handle this non-convexity, we adopt SCA framework to transform $\mathcal{P}2$ into a convex one. We introduce the set of positive auxiliary variables $\rho_k$ and $\varpi_k,~\forall~k\in\mathcal{K}$ such that:
\begin{equation}
    \frac{\left| \mathbf{h}_k^H \mathbf{G} \mathbf{v}_k\right|^2}{\varpi_k} \ge \rho_k. \label{Eq:18}
\end{equation}
and
\begin{equation}
    \sum_{i = 1, i\neq k}^K \left| \mathbf{h}_k^H \mathbf{G} \mathbf{v}_i\right|^2 + \sigma_n^2 \le \varpi_k. \label{Eq:19}
\end{equation}

Although~\eqref{Eq:18} is still non-convex, it is apparent that~\eqref{Eq:19} is convex with respect to $\mathbf{v}_k$. Since the left-hand side of~\eqref{Eq:18} is lower-bounded by its first-order Taylor expansion around an initial point $\left(\mathbf{v}_k^{\left(m\right)},\varpi_k^{\left(m\right)}\right)$ in the $m$-th iteration, we get
\begin{align}
    \frac{\left| \mathbf{h}_k^H \mathbf{G} \mathbf{v}_k\right|^2}{\varpi_k} &\ge -\left( \frac{\left| \mathbf{h}_k^H \mathbf{G} \mathbf{v}_k^{\left(m\right)}\right|}{\varpi_k^{\left(m\right)}}\right)^2 \varpi_k \nonumber \\ &+ \frac{2\mathcal{R}\left[\left(\mathbf{h}_k^H \mathbf{G} \mathbf{v}_k^{\left(m\right)}\right)^H \left(\mathbf{h}_k^H \mathbf{G} \mathbf{v}_k\right)\right]}{\varpi_k^{\left(m\right)}}  \equiv {\Lambda_k^{\text{Taylor}}}^{\left(m\right)}. \label{Eq:20}
\end{align}

Now, the right-hand side of~\eqref{Eq:20} is an affine function that can be used to approximate the left-hand side of~\eqref{Eq:18}. Regarding constraint~\eqref{Eq:16e}, it can be rewritten as:
\begin{equation}
    \frac{\delta_{1,u}}{\delta_{0,u}} - \ln{\frac{\delta_{1,u}}{\delta_{0,u}}} - 1 \le \frac{2\epsilon^2}{J}. \label{Eq:21}
\end{equation}

We introduce the positive set of auxiliary variables $\tau_u,~\forall~u\in\mathcal{U}$, such that $\delta_{1,u} \le \tau_u \delta_{0,u}$. Hence, problem $\mathcal{P}2$ can be reformulated in the $m$-th iteration as:
\begin{subequations}
\begin{align}
	\mathcal{P}3: \max_{\left\{\mathbf{v}_k,\rho_k,\varpi_k\right\}_{k=1}^{K},\left\{ \tau_u\right\}_{u=1}^U}~~\sum_{k=1}^{K} \log_2 \left(1+\rho_k\right), \label{Eq:22a} \\
	& \hspace{-2.5in} \textrm{Subject to:} \nonumber \\
	& \hspace{-2in} \rho_k \ge \gamma^{\text{min}},~\forall~k\in \mathcal{K}, \label{Eq:22b}\\
    & \hspace{-2in} {\Lambda_k^{\text{Taylor}}}^{\left(m\right)} \ge \rho_k,~\forall~k\in \mathcal{K}, \label{Eq:22c}\\
    & \hspace{-2in} \sum_{i = 1, i\neq k}^K \left| \mathbf{h}_k^H \mathbf{G} \mathbf{v}_i\right|^2 + \sigma_n^2 \le \varpi_k,~\forall~k\in \mathcal{K}, \label{Eq:22d}\\
    & \hspace{-2in} \sum_{i = 1}^K |\mathbf{h}_u^H \mathbf{G} \mathbf{v}_i|^2 \le \sigma_u^2\left(\tau_u - 1 \right),~\forall~u\in \mathcal{U}, \label{Eq:22e}\\
    & \hspace{-2in} \tau_u - \ln{\tau_u} - 1 \le \frac{2\epsilon^2}{J},~\forall~u\in \mathcal{U}, \label{Eq:22f}\\
    & \hspace{-2in} \eqref{Eq:16c}, \label{Eq:22g}
\end{align}
\label{Eq:22}
\end{subequations}
where $\gamma^{\text{min}} = 2^{R_k^{\text{min}}} - 1$. Now, problem $\mathcal{P}3$ is convex and can be directly solved using standard convex optimization tools, such as CVX toolbox~\cite{grant2009cvx}.
\subsection{SIM Phase Shift Optimization}
Assuming fixed beamfocusing vectors $\left\{\mathbf{v}_k \right\}_{k=1}^{K}$, the SIM phase shift optimization problem can be formulated as:
\begin{subequations}
\begin{align}
	\mathcal{P}4: \max_{\left\{\mathbf{\Theta}^l\right\}_{l=1}^{L}} \sum_{k=1}^{K} \log_2 \left(1+\gamma_k\right), \label{Eq:23a} \\
	& \hspace{-1.5in} \textrm{Subject to:} \nonumber \\
	& \hspace{-1in} \eqref{Eq:16b},~\eqref{Eq:16d},~\text{and}~\eqref{Eq:16e}. \label{Eq:23b}
\end{align}
\label{Eq:23}
\end{subequations}

Problem $\mathcal{P}4$ is challenging due to the coupling between the phase shift matrices among all SIM layers in the objective function~\eqref{Eq:23a} as well as the constraints~\eqref{Eq:16b}, ~\eqref{Eq:16d}, and~\eqref{Eq:16e}. To solve this problem, we will first adopt SCA to solve the problem layer-by-layer. Then, we will solve the same problem by adopting the PGA algorithm~\cite{An2023stacked}. 
\subsubsection{Per Layer SCA-based Phase Shift Design}
In this scheme, we are going to solve the problem layer-by-layer, where to optimize the phase shift of the $l$-th layer, the phase shift of other layers in the set $\mathcal{L}/\left\{ l\right\}$ will be considered fixed. Consequently, we have to separate $\mathbf{\Theta}^l$ from the matrix $\mathbf{G}$. Denote $\mathbf{G} = \mathbf{G}_L^l \mathbf{\Theta}^{l}\mathbf{G}_R^l$, where $\mathbf{G}_L^l$ and $\mathbf{G}_R^l$ are given by:
    \begin{equation}
        \mathbf{G}_L^l = 
\begin{cases}
\mathbf{\Theta}^L \mathbf{W}^L \dots \mathbf{\Theta}^{l+1} \mathbf{W}^{l+1}, & \text{if } l \neq L, \\
\mathbf{I}_N, & \text{if } l = L,
\end{cases} \label{Eq:31}
    \end{equation}
    \begin{equation}
        \mathbf{G}_R^l = 
\begin{cases}
\mathbf{W}^l\mathbf{\Theta}^{l-1}\mathbf{W}^{l-1} \dots \mathbf{\Theta}^{1} \mathbf{W}^{1}, & \text{if } l \neq 1, \\
\mathbf{W}^{1}, & \text{if } l = 1.
\end{cases} \label{Eq:32}
    \end{equation}

As we did in the beamfocusing problem, to handle the fractional SINR term in the objective~\eqref{Eq:23a}, we introduce the set of auxiliary variables $\varsigma_k$ and $\upsilon_k,~\forall k \in \mathcal{K}$ such that:
\begin{equation}
    \frac{\left| \pmb{\phi}_l^T \tilde{\mathbf{h}}_k\right|^2}{\varsigma_k} \ge \upsilon_k, \label{Eq:18-1}
\end{equation}
and
\begin{equation}
    \sum_{i = 1, i\neq k}^K \left| \pmb{\phi}_l^T \tilde{\mathbf{h}}_i\right|^2 + \sigma_n^2 \le \varsigma_k, \label{Eq:19-1}
\end{equation}
where $\mathbf{h}_k^H \mathbf{G} \mathbf{v}_k=\pmb{\phi}_l^T \tilde{\mathbf{h}}_k$ $\mathbf{h}_k^H \mathbf{G} \mathbf{v}_i=\pmb{\phi}_l^T \tilde{\mathbf{h}}_i$, where  $\tilde{\mathbf{h}}_k=\text{diag}\left( \mathbf{h}_k^H \mathbf{G}_L^l\right) \mathbf{G}_R^l\mathbf{v}_k$ and $\tilde{\mathbf{h}}_i=\text{diag}\left( \mathbf{h}_k^H \mathbf{G}_L^l\right) \mathbf{G}_R^l\mathbf{v}_i$. In addition, $\pmb{\phi}_l = \left[ e^{j\theta_1^l}, e^{j\theta_2^l}, \dots e^{j\theta_N^l}\right]^T$ is the column vector that comprises the phase shift of the $l$-th layers. Equation~\eqref{Eq:19-1} is convex with respect to $\pmb{\phi}_l$. However,~\eqref{Eq:18-1} is still non-convex. Since the left-hand side of~\eqref{Eq:18-1} is lower-bounded by its first-order Taylor expansion around an initial point $\left(\pmb{\phi}_l^{\left(m\right)},\varsigma_k^{\left(m\right)}\right)$ at the $m$-th iteration, we get
\begin{align}
    \frac{\left| \pmb{\phi}_l^T \tilde{\mathbf{h}}_k\right|^2}{\varsigma_k} &\ge -\left( \frac{\left| \left( \pmb{\phi}_l^T\right)^{\left(m\right)} \tilde{\mathbf{h}}_k\right|}{\varsigma_k^{\left(m\right)}}\right)^2 \varsigma_k \nonumber \\ &+ \frac{2\mathcal{R}\left[\left(\left( \pmb{\phi}_l^T\right)^{\left(m\right)} \tilde{\mathbf{h}}_k\right)^H \left( \pmb{\phi}_l^T \tilde{\mathbf{h}}_k\right)\right]}{\varsigma_k^{\left(m\right)}}  \equiv {\Xi_k^{\text{Taylor}}}^{\left(m\right)}. \label{Eq:20-1}
\end{align}

Now, the right-hand side of~\eqref{Eq:20-1} is an affine function that can be used to approximate the left-hand side of~\eqref{Eq:18-1}. The constraint~\eqref{Eq:16d} can be relaxed to $\left| \theta_n^l\right| \le 1,~\forall~n \in \mathcal{N}, l\in \mathcal{L}$. Regarding the constraint~\eqref{Eq:16e}, it can be handled in a similar manner as in the beamfocusing sub-problem by introducing the positive set of auxiliary variables $\zeta_u,~\forall~u\in\mathcal{U}$, such that $\delta_{1,u} \le \zeta_u \delta_{0,u}$. Hence, problem $\mathcal{P}4$ can be reformulated in the $m$-th iteration as:
\begin{subequations}
\begin{align}
	\mathcal{P}5: \max_{\left\{\pmb{\phi}_l\right\}_{l=1}^{L},\left\{\upsilon_k,\varsigma_k\right\}_{k=1}^{K},\left\{ \zeta_u\right\}_{u=1}^U}~~\sum_{k=1}^{K} \log_2 \left(1+\upsilon_k\right), \label{Eq:22-1a} \\
	& \hspace{-2.5in} \textrm{Subject to:} \nonumber \\
	& \hspace{-2in} \upsilon_k \ge \gamma^{\text{min}},~\forall~k\in \mathcal{K}, \label{Eq:22-1b}\\
    & \hspace{-2in} {\Xi_k^{\text{Taylor}}}^{\left(m\right)} \ge \upsilon_k,~\forall~k\in \mathcal{K}, \label{Eq:22-1c}\\
    & \hspace{-2in} \sum_{i = 1, i\neq k}^K \left| \pmb{\phi}_l^T \tilde{\mathbf{h}}_i\right|^2 + \sigma_n^2 \le \varsigma_k,~\forall~k\in \mathcal{K}, \label{Eq:22-1d}\\
    & \hspace{-2in} \sum_{i = 1}^K \left|\pmb{\phi}_l^T \tilde{\mathbf{h}}_u\right|^2 \le \sigma_u^2\left(\zeta_u - 1 \right),\forall~u\in \mathcal{U}, \label{Eq:22-1e}\\
    & \hspace{-2in} \zeta_u - \ln{\zeta_u} - 1 \le \frac{2\epsilon^2}{J},~\forall~u\in \mathcal{U}, \label{Eq:22-1f}\\
    & \hspace{-2in} \left| \theta_n^l\right| \le 1,~\forall~n \in \mathcal{N}, l\in \mathcal{L}, \label{Eq:22-1g}
\end{align}
\label{Eq:22-1}
\end{subequations}
where $\mathbf{h}_u^H \mathbf{G} \mathbf{v}_i=\pmb{\phi}_l^T \tilde{\mathbf{h}}_u$, where $\tilde{\mathbf{h}}_u = \text{diag}\left( \mathbf{h}_u^H \mathbf{G}_L^l\right) \mathbf{G}_R^l\mathbf{v}_i$. Now, problem $\mathcal{P}5$ is convex and can be directly solved using standard convex optimization tools, such as CVX toolbox~\cite{grant2009cvx}. After solving the problem, the obtained phase shifts can be projected back onto the unit modulus set.
\subsubsection{PGA-based Phase Shift Design}

First, to facilitate handling $\mathcal{P}4$, we transform constraints ~\eqref{Eq:16b} and~\eqref{Eq:16e} as penalty terms in the objective function to obtain the following problem:
\begin{subequations}
\begin{align}
	\mathcal{P}6: \max_{\left\{\mathbf{\Theta}^l\right\}_{l=1}^{L}} F\left( \theta_n^l\right) =  R\left( \theta_n^l\right) + \mu_1 A\left( \theta_n^l\right) - \mu_2 B\left( \theta_n^l\right), \label{Eq:24a} \\
	& \hspace{-3in} \textrm{Subject to:} \nonumber \\
	& \hspace{-2.8in} \left| \theta_n^l\right| = 1,~\forall~n \in \mathcal{N}, l\in \mathcal{L}, \label{Eq:24b}
\end{align}
\label{Eq:24}
\end{subequations}
where $\mu_1,~\mu_2 \in \mathbb{R}^+$ are the penalty factors that are used to force compliance of the constraints. In addition, $$A\left( \theta_n^l\right) = \sum_{k=1}^K \min\left(\gamma_k\left( \theta_n^l\right)-\gamma^{\text{min}},0\right)^2,$$ and \small $$B\left( \theta_n^l\right) = \sum_{u=1}^U \max\left(\frac{\delta_{1,u}\left( \theta_n^l\right)}{\delta_{0,u}} - \ln{\frac{\delta_{1,u}\left( \theta_n^l\right)}{\delta_{0,u}}} - 1 - \frac{2\epsilon^2}{J},0\right)^2.$$ \normalsize

Now, the structure of $\mathcal{P}6$ makes the PGA method a suitable choice to solve it, since the elements of $\mathbf{\Theta}^l$ lie on the unit circle in the complex plane. PGA can be applied to handle $\mathcal{P}6$ according to the following steps:

\textbf{\textit{Step 1 (Initialize the SIM Phase Shifts)}:} In this part, we randomly initialize the SIM phase shifts $\theta_n^l \in \left[ 0,2\pi\right)$ such that $\left| \theta_n^l\right| = 1,~\forall~n\in\mathcal{N},~l\in\mathcal{L}$.

\textbf{\textit{Step 2 (Calculate the Partial Derivatives)}:} We calculate the partial derivatives of the objective function $F$ with respect to the phase shift of the $n$-th meta atom in the $l$-th layer, i.e., $\theta_n^l,~\forall~n\in\mathcal{N},~l\in\mathcal{L}$ as follows:
\begin{equation}
    \frac{\partial F\left( \theta_n^l\right)}{\partial \theta_n^l} = \frac{\partial R\left( \theta_n^l\right)}{\partial \theta_n^l} + \mu_1 \frac{\partial A\left( \theta_n^l\right)}{\partial \theta_n^l} - \mu_2 \frac{\partial B\left( \theta_n^l\right)}{\partial \theta_n^l}. \label{Eq:25}
\end{equation}

\textbf{\textit{Step 3 (Normalize the Partial Derivatives)}:} After calculating all the partial derivatives of the objective function $F$ with respect to $\theta_n^l$, we normalize the partial derivatives to mitigate oscillation during optimization as follows~\cite{niu2024stacked}:
\begin{equation}
    \frac{\partial F\left( \theta_n^l\right)}{\partial \theta_n^l} \leftarrow \frac{\pi}{\kappa} \frac{\partial F\left( \theta_n^l\right)}{\partial \theta_n^l}, \label{Eq:25-1}
\end{equation}
where $\kappa = \max\left(\frac{\partial F\left( \theta_n^l\right)}{\partial \theta_n^l}\right), \forall n\in \mathcal{N}, l \in \mathcal{L}$.

\textbf{\textit{Step 4 (Update the SIM Phase Shifts)}:} After calculating all the partial derivatives of the objective function $F$ with respect to $\theta_n^l$, we update the phase shifts as follows:
\begin{equation}
    \theta_n^l \leftarrow \theta_n^l + \alpha  \frac{\partial F\left( \theta_n^l\right)}{\partial \theta_n^l},\label{Eq:26}
\end{equation}
where $\alpha$ denotes the Armijo step size, which is updated at each iteration according to the backtracking line search procedure~\cite{An2023stacked,niu2024stacked}.

\textbf{\textit{Step 5 (Iterate Until Convergence)}:}
The last step is to repeat steps 2 and 3 until the improvement in the objective function $F$ falls below a given threshold.

To calculate the partial derivatives of the objective function $F$ in~\eqref{Eq:25}, we start by calculating the partial derivatives of the sum covert rate $R\left( \theta_n^l\right)$. Applying the chain rule, we get:
\begin{equation}
    \frac{\partial R\left( \theta_n^l\right)}{\partial \theta_n^l} = \frac{1}{\ln{2}} \sum_{k=1}^K \frac{1}{1+ \gamma_k\left( \theta_n^l\right)} \frac{\partial \gamma_k\left( \theta_n^l\right)}{\partial \theta_n^l}. \label{Eq:27}
\end{equation}

\begin{thm}
    The partial derivatives of $\frac{\partial \gamma_k}{\partial \theta_n^l}$ is given by:
    \begin{equation}
        \frac{\partial \gamma_k\left( \theta_n^l\right)}{\partial \theta_n^l} = 2\varrho_k \left( \gamma_k\left( \theta_n^l\right) \left(\sum_{i = 1, i\neq k}^K \psi_{k,i} \right) - \psi_{k,k} \right), \label{Eq:28}
    \end{equation}
    where
    \begin{equation}
        \varrho_k = \frac{1}{\sum_{i = 1, i\neq k}^K \left| \mathbf{h}_k^H \mathbf{G} \mathbf{v}_i\right|^2 + \sigma_n^2}, \label{Eq:29}
    \end{equation}
    and
    \begin{equation}
        \psi_{k,i} = \mathcal{I} \left[e^{j \theta_n^l} \left(\mathbf{h}_k^H \mathbf{G} \mathbf{v}_i\right)^* \mathbf{h}_k^H \mathbf{G}_L^l \mathbf{E}_{n,n} \mathbf{G}_R^l \mathbf{v}_i\right], \label{Eq:30}
    \end{equation} 
    \label{thm1}
\end{thm}
where $\mathbf{E}_{n,n}$ is an $N\times N$ matrix whose entries are all zero except for the $(n,n)$-th element, which is equal to 1.
\proof See the Appendix.

Using the result of Theorem~\ref{thm1}, the partial derivative of the second term on the right-hand side of~\eqref{Eq:25} is given as:
\begin{equation}
    \frac{\partial A\left( \theta_n^l\right)}{\partial \theta_n^l} =2\sum_{k=1}^K 1_{\left[\gamma_k<\gamma^{\text{min}}\right]} \left(\gamma_k\left( \theta_n^l\right) - \gamma^{\text{min}}\right) \frac{\partial \gamma_k\left( \theta_n^l\right)}{\partial \theta_n^l}, \label{Eq:33}
\end{equation}
where $1_{\left[\gamma_k<\gamma^{\text{min}}\right]}$ denotes the indicator function that equals 1 to penalize the objective function when constraint~\eqref{Eq:16b} violates and 0 otherwise. 

To calculate the partial derivative of the last term on the right-hand side of~\eqref{Eq:25}, denote
\begin{dmath}
    g_u\left( \theta_n^l\right) = \frac{\delta_{1,u}\left( \theta_n^l\right)}{\delta_{0,u}} - \ln{\frac{\delta_{1,u}\left( \theta_n^l\right)}{\delta_{0,u}}} - 1 = z_u\left( \theta_n^l\right)-\ln{z_u\left( \theta_n^l\right)} -1, \label{Eq:34}
\end{dmath}
where $z_u\left( \theta_n^l\right) = \frac{\sum_{k = 1}^K |\mathbf{h}_u^H \mathbf{G} \mathbf{v}_k|^2 + \sigma_u^2}{\sigma_u^2}$. 

Now, we can write
\begin{equation}
    \frac{\partial B\left( \theta_n^l\right)}{\partial \theta_n^l} =2\sum_{u=1}^U 1_{\left[g_u>\frac{2\epsilon^2}{J}\right]} \left(g_u\left( \theta_n^l\right) - \frac{2\epsilon^2}{J}\right) \frac{\partial g_u\left( \theta_n^l\right)}{\partial \theta_n^l}, \label{Eq:35}
\end{equation}
where, following similar procedures as in Appendix~\ref{App1},
\begin{dmath}
    \frac{\partial g_u\left( \theta_n^l\right)}{\partial \theta_n^l} = \left(1-\frac{1}{z_u\left( \theta_n^l\right)}\right)\frac{\partial z_u\left( \theta_n^l\right)}{\partial \theta_n^l} = \frac{-2}{\sigma_u^2}\left(1-\frac{1}{z_u\left( \theta_n^l\right)}\right) \times \sum_{k = 1}^K \mathcal{I}\left[e^{j \theta_n^l} \left(\mathbf{h}_u^H \mathbf{G} \mathbf{v}_k\right)^* \mathbf{h}_u^H \mathbf{G}_L^l \mathbf{E}_{n,n} \mathbf{G}_R^l \mathbf{v}_k\right]. \label{Eq:36}
\end{dmath}

Based on \eqref{Eq:27}, \eqref{Eq:33}, and \eqref{Eq:35}, the gradient of the objective function $F$ with respect to the phase shift of the $n$-th meta-atom in the $l$-th layer can be derived, and is given in \eqref{Eq:37} at the top of the next page.
\begin{figure*}[!t]
\begin{dmath}
    \frac{\partial F\left( \theta_n^l\right)}{\partial \theta_n^l} = \sum_{k=1}^K \frac{\partial \gamma_k\left( \theta_n^l\right)}{\partial \theta_n^l}\left(\frac{1}{\ln{2}\left(1+ \gamma_k\left( \theta_n^l\right)\right)} + 2\mu_1 \times 1_{\left[\gamma_k<\gamma^{\text{min}}\right]}  \left(\gamma_k\left( \theta_n^l\right) - \gamma^{\text{min}}\right)\right) -2 \mu_2 \sum_{u=1}^U \left(1_{\left[g_u>\frac{2\epsilon^2}{J}\right]} \left(g_u\left( \theta_n^l\right) - \frac{2\epsilon^2}{J}\right) \frac{\partial g_u\left( \theta_n^l\right)}{\partial \theta_n^l}\right). \label{Eq:37}
\end{dmath}
\vspace*{2pt}
\hrulefill
\end{figure*}
\subsection{Overall Algorithm and Complexity Analysis}

Based on the previously formulated sub-problems and their corresponding solution methods, the overall procedure of the two proposed AO algorithms is summarized in Algorithm~\ref{alg1} for the SCA-based phase shift design and in Algorithm~\ref{alg2} for the PGA-based phase shift design. 
\begin{algorithm}
\caption{AO algorithm with SCA-based phase shift design}\label{alg1}
\begin{algorithmic}[1]
\STATE \textbf{Input:} Channel vectors $\mathbf{h}_k,\forall k \in \mathcal{K}$, $\mathbf{h}_u,\forall u \in \mathcal{U}$, SIM propagation coefficient matrices $\mathbf{W}^l,\forall l \in \mathcal{L}$, and noise variance $\sigma_n^2$.
\STATE \textbf{Initialize:} $\mathbf{v}_k^{\left(0\right)},\forall k \in \mathcal{K},{\mathbf{\Theta}^l}^{\left(0\right)},\forall l \in \mathcal{L}$ and the auxiliary variables $\varpi_k^{\left(0\right)},\varsigma_k^{\left(0\right)},\forall k \in \mathcal{K}$.
\STATE Set $m \leftarrow 1$.
\REPEAT
    \STATE \textbf{Beamfocusing optimization:}
	\STATE \hspace{0.4cm} Solve problem $\mathcal{P}3$ to design $\mathbf{v}_k^{\left(m\right)},\forall k \in \mathcal{K}$ using $\mathbf{v}_k^{\left(m-1\right)}, {\mathbf{\Theta}^l}^{\left(m-1\right)},\forall l \in \mathcal{L}$.
    \STATE \textbf{Phase-shift optimization:}
	\STATE \hspace{0.4cm} Solve problem $\mathcal{P}5$ to design ${\mathbf{\Theta}^l}^{\left(m\right)},\forall l \in \mathcal{L}$ using the designed $\mathbf{v}_k^{\left(m\right)}$ and ${\mathbf{\Theta}^l}^{\left(m-1\right)},\forall l \in \mathcal{L}$.
    \STATE Update all the optimization and auxiliary variables.
	\STATE $m \leftarrow m + 1$.
\UNTIL Convergence.
\STATE \textbf{Output} Optimal beamformers $\mathbf{v}_k^{\ast},\forall k\in\mathcal{K}$ and phase-shift matrices ${\boldsymbol{\Theta}^l}^{\ast},\forall l\in\mathcal{L}$.
\end{algorithmic}
\end{algorithm}

\begin{algorithm}[t]
\caption{AO algorithm with PGA-based phase shift design}
\label{alg2}
\begin{algorithmic}[1]

\STATE \textbf{Input:} Channel vectors $\mathbf{h}_k,\forall k \in \mathcal{K}$, $\mathbf{h}_u,\forall u \in \mathcal{U}$, SIM propagation coefficient matrices $\mathbf{W}^l,\forall l \in \mathcal{L}$, noise variance $\sigma_n^2$, and penalty factors $\mu_1, \mu_2$.
\STATE \textbf{Initialize:} $\mathbf{v}_k^{\left(0\right)},\forall k \in \mathcal{K}$ and ${\mathbf{\Theta}^l}^{\left(0\right)},\forall l \in \mathcal{L}$.
\STATE Set $m \leftarrow 1$.

\REPEAT
    \STATE Calculate the partial derivatives of the objective function $F$ with respect to ${\theta_n^l}^{\left(m\right)},~\forall~n\in\mathcal{N},~l\in\mathcal{L}$ using~\eqref{Eq:37}.
    \STATE Normalize the partial derivatives using~\eqref{Eq:25-1}.
    \STATE \textbf{Initialize:} Armijo step size $\alpha^{0}=\alpha_0>0$, the factor $\varepsilon \in \left(0,1\right)$, and $t \leftarrow 1$.
    \REPEAT
        \STATE Update ${\mathbf{\Theta}^l}^{\left(t\right)},\forall l \in \mathcal{L}$ using~\eqref{Eq:26}.
        \STATE Solve problem $\mathcal{P}3$ to obtain $\mathbf{v}_k^{(t)},\forall k\in\mathcal{K}$ using ${\mathbf{\Theta}^l}^{(t)},\forall l\in\mathcal{L}$.
        \IF{$F^{(t)}<F^{(t-1)}$}
        \STATE $\alpha^{(t)} \leftarrow \varepsilon \alpha^{(t-1)}$
        \ENDIF
        \STATE $t \leftarrow t+1$
    \UNTIL $F^{(t)} \ge F^{(t-1)}$.
    \STATE Update optimization variables.
    \STATE $m \leftarrow m + 1$.
\UNTIL Convergence
\STATE \textbf{Output:} Optimal beamformers $\mathbf{v}_k^{\ast},\forall k\in\mathcal{K}$ and phase-shift matrices ${\boldsymbol{\Theta}^l}^{\ast},\forall l\in\mathcal{L}$.
\end{algorithmic}
\end{algorithm}

As shown in Algorithm~\ref{alg1}, the complexity of the overall algorithm is dominated by the complexity of solving problems $\mathcal{P}3$ and $\mathcal{P}5$. The transmit beamfocusing problem $\mathcal{P}3$ and the SCA-based phase shift problem $\mathcal{P}5$, are formulated as second-order cone programming (SOCP) problems, which are solved through the CVX's interior point method solvers. The complexities of these problems are $\mathcal{O}\left(\sqrt{K+U}\left( KM\right)^{3}\right)$ and $\mathcal{O}\left( \sqrt{K+U+N}~L N^{3}\right)$, respectively. Consequently, the overall computational complexity is in the order of $\mathcal{O}\left(Q\left(\sqrt{K+U}\left(KM\right)^{3} + \sqrt{K+U+N}~LN^3\right)\right)$, where $Q$ is the number of iterations in Algorithm~\ref{alg1}. Regarding Algorithm~\ref{alg2}, the complexity of the PGA-based phase shift design method is $\mathcal{O}\left(K L N\right)$. Hence, the overall computational complexity of algorithm~\ref{alg2} is $\mathcal{O}\left(Q\left(T\sqrt{K+U}\left(KM\right)^{3} + KLN\right)\right)$, where $T$ is the number of iterations in the inner Armijo step size loop.
\section{Simulation Results}
\label{results}

\begin{table*}[t]
	\renewcommand{\arraystretch}{1.5}
        \centering
        \caption{Simulation Parameters}
	\begin{tabularx}{\textwidth}{Xc||Xc}
		\hline\hline
		\textbf{Parameter} & \textbf{Value} & \textbf{Parameter} & \textbf{Value} \\
		\hline
        Number of Tx antennas, $M$ & $6$ & Number of covert UEs, $K$  & $3$\\
		Noise power spectral density & $-174$ dBm/Hz & Number of wardens, $U$ & $3$\\
		Number of SIM layers, $L$ & $5$ & Number of meta atoms, $N$ & $45$ \\
		Carrier frequency & $10$ GHz & Communication bandwidth, $B$ & $1$ MHz \\
		BS height, $Z$  & $5$ m& Radius of service area & $10$ m\\
		Maximum allowed transmit power $P_{\text{max}}$ & $40$ dBm & Covertness threshold $\epsilon$  & $0.1$\\
		Users' rate requirement, $R_k^{\text{min}}$ & $0.1$ bps/Hz &  Number of Warden's observations per time slot, $J$ & $10$\\
        Meta atoms spacing, $dx$, $dz$ &$\lambda$ m & Area of each meta atom & $\lambda^2/4$ m$^2$\\
		\hline
	\end{tabularx}
	\label{tab1}
\end{table*}

In this section, we conduct extensive numerical simulations to evaluate the performance of the proposed SIM-assisted near-field covert communication system. The BS is located at the point $\left(0,0,Z\right)$ of a three-dimensional Cartesian coordinate system, where $Z$ is the BS height, and all communication UEs and wardens are randomly located in the $x-y$ plane according to a uniform distribution within a circle of radius $10$ m centered at $\left(10,10,0\right)$ m to comply with the near-field Rayleigh distance conditions. The thickness of the SIM is set to $5\lambda$. Hence, the spacing between any two layers and between the BS and layer 1 is $d_{\text{SIM}}=5\lambda/L$. The path loss exponent is set to $\eta=2.5$. Unless otherwise stated, the simulation parameters are set as shown in Table~\ref{tab1}. The propagation distance, $r_{n,\tilde{n}}^l$, from the $\tilde{n}$-th meta atom in the $l-1$-th layer to the $n$-th meta atom in the $l$-th layer is given by:

\small
\begin{equation}
    r_{n,\tilde{n}}^l = \sqrt{d_{\text{SIM}}^2+ \lambda^2 \left( \left \lfloor \frac{\left| n-\tilde{n}\right|}{N_{\text{max}}} \right \rfloor^2+\left[\text{mod}\left(\left| n-\tilde{n}\right|,N_{\text{max}}\right)\right]^2\right)},
\end{equation}
\normalsize
where $N_{\text{max}}=\max\left\{N_x,N_z\right\}$. Additionally, the propagation distance, $r_{n,m}^1$, from the $m$-th transmit antenna to the $n$-th meta atom in the $1^{st}$ layer is given by~\eqref{Eq:45} in the top of the next page. All simulation results are averaged over $100$ independent channel realizations drawn from the adopted near-field channel model.
\begin{figure*}[!t]
\begin{equation}
     r_{n,m}^1 = \sqrt{d_{\text{SIM}}^2+ \lambda^2 \left[ \left(\text{mod}\left( n-1,N_{\text{max}}\right)-\frac{1+N_{\text{max}}}{2}\right)-\left (m- \frac{1+M}{2} \right )\right]^2 + \lambda^2 \left(\left \lceil \frac{n}{N_{\text{max}}} \right \rceil - \frac{1+N_{\text{max}}}{2} \right)^2} \label{Eq:45}
\end{equation}
\vspace*{2pt}
\hrulefill
\end{figure*}

For comparison purposes, we are going to compare our proposed SCA-based and PGA-based algorithms with two benchmarks; 1) \textbf{Random phase shift}: in which the beamfocusing vectors are designed by solving $\mathcal{P}3$, while the phase shift matrices of all layers are chosen randomly to comply with the unit modulus constraint for the phase shift of each meta atom. 2) \textbf{Codebook-based phase shift}: in which we generate a codebook for the phase shift vectors that comprises $100$ phase shift matrices each of size $\left(L \times N\right)$. For each matrix, the beamfocusing vectors are designed by solving $\mathcal{P}3$, and the phase shift matrix that gives the maximum achieved sum covert rate is selected.

The convergence performance of the proposed SCA-based phase shift and PGA-based phase shift algorithms is presented in Fig.~\ref{res1} at different number of SIM layers assuming $M=4$ transmit antennas. As we can see in the figure, both algorithms converge to a stationary point in a relatively small number of iterations, most of the time less than $20$ iterations. In addition, the sum covert rate achieved by adopting the proposed SCA-based phase shift algorithm is higher than that achieved by adopting the PGA-based phase shift algorithm for all values of $L$. For example, the SCA-based phase shift algorithm outperforms the PGA-based phase shift algorithm with almost $225$\% at $L=5$ layers. Moreover, for both algorithms, the achieved sum covert rate increases with increasing the number of layers, where the achieved rate by the SCA-based phase shift algorithm at $L=5$ layers is higher with almost $17.2$\% that that achieved by the same algorithm adopting only single layer. This can be attributed to the fact that increasing the number of SIM layers provides more opportunity to configure the wireless channel to reduce the inter-user interference. However, the further increase of the number of SIM layers may cause degradation of the achieved covert rate as a result of the multiplication of the propagation coefficient matrices in the effective channel from the BS to the UEs.

\begin{figure}[!t]
\centerline{\includegraphics[width=0.5\textwidth]{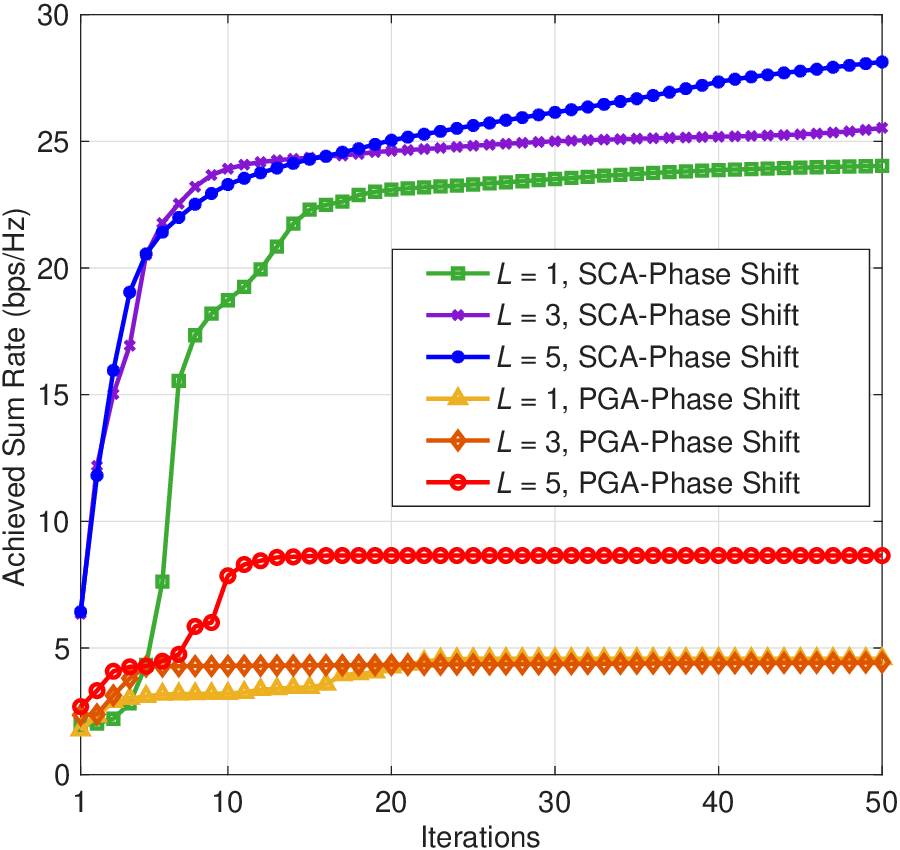}}
\caption{Convergence performance of the proposed system for different number of Layers.}
\label{res1}
\end{figure}
\begin{figure}[!t]
\centerline{\includegraphics[width=0.5\textwidth]{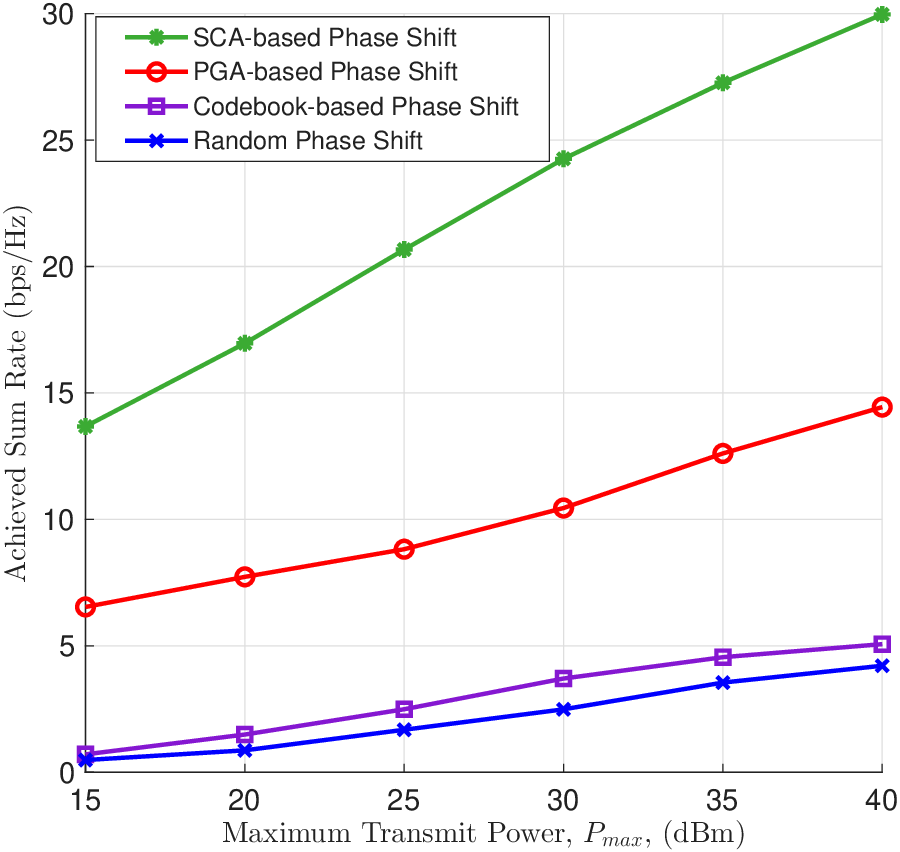}}
\caption{Sum rate performance of the proposed system against the maximum transmit power as compared with benchmarks.}
\label{res2}
\end{figure}

In Fig.~\ref{res2}, the sum covert rate performance of the proposed system is studied as a function of the maximum allowed transmit power assuming $M=4$ transmit antennas. The achieved covert rate is compared with the PGA-based phase shift algorithm as well as with the benchmarks mentioned above. As shown in the figure, as the maximum transmit power increases, the achieved sum covert rate increases for all phase shift design schemes. In addition, the proposed SCA-based algorithm achieves the highest sum covert rate for all values of maximum transmit power. For example, there is an average increase in the sum rate of almost 120\% in favor of the SCA-based method compared to the PGA-method. Moreover, the achieved sum rate by the PGA-based method is higher than that achieved by the codebook-based phase shift or the random phase shift by more than 300\%. We can also observe that the codebook-based phase shift and the random phase shift schemes achieve very close performance over the full transmit power range. This is because both approaches do not perform SIM phase shift optimization; instead, the random phase shift scheme applies completely random phase shift matrices for all layers while the codebook-based phase shift selects non-optimized phase configurations from a random pool. Consequently, both schemes do not effectively exploit the beamfocusing capability of the SIM, leading to comparable and relatively limited sum rate performance.

The effect of increasing the number of transmit antennas on the achieved sum rate of the proposed system is presented in Fig.~\ref{res3}. The figure shows that as the number of transmit antennas increases, the achieved sum rate increases for all the schemes. Additionally, the proposed SCA-based phase shift algorithm outperforms other algorithms over the entire range of transmit antennas, where the achieved sum rate with the SCA-based phase shift design scheme is almost $100$\% higher than that achieved with the PGA-based phase shift algorithm and almost $173$\% higher than that achieved with the codebook-based and random phase shift schemes at $M=8$ antennas. Moreover, we can notice that the PGA-based phase shift algorithm performs better with a higher number of transmit antennas as compared to the codebook-based phase shift and random phase shift algorithms. More specifically, the PGA-based method seems to provide poor performance at $M<6$. Consequently, we will set $M=6$ transmit antennas in all subsequent results. 
\begin{figure}[!t]
\centerline{\includegraphics[width=0.5\textwidth]{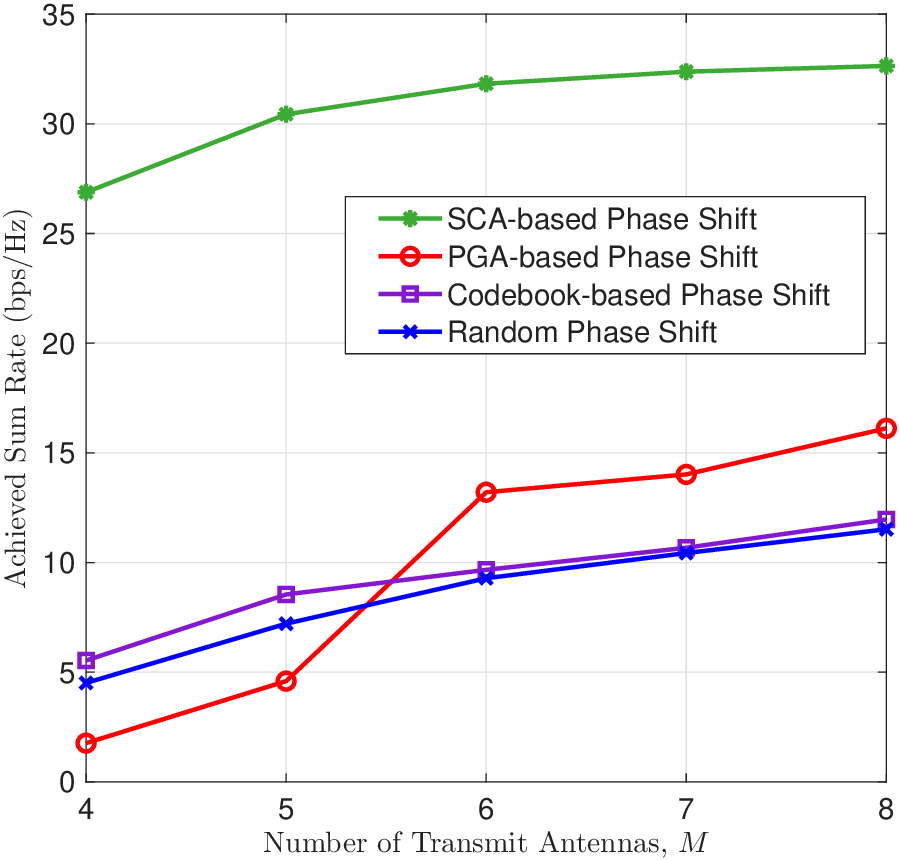}}
\caption{Sum rate performance of the proposed system against the number of transmit antennas as compared with benchmarks.}
\label{res3}
\end{figure}

\begin{figure}[!t]
\centerline{\includegraphics[width=0.5\textwidth]{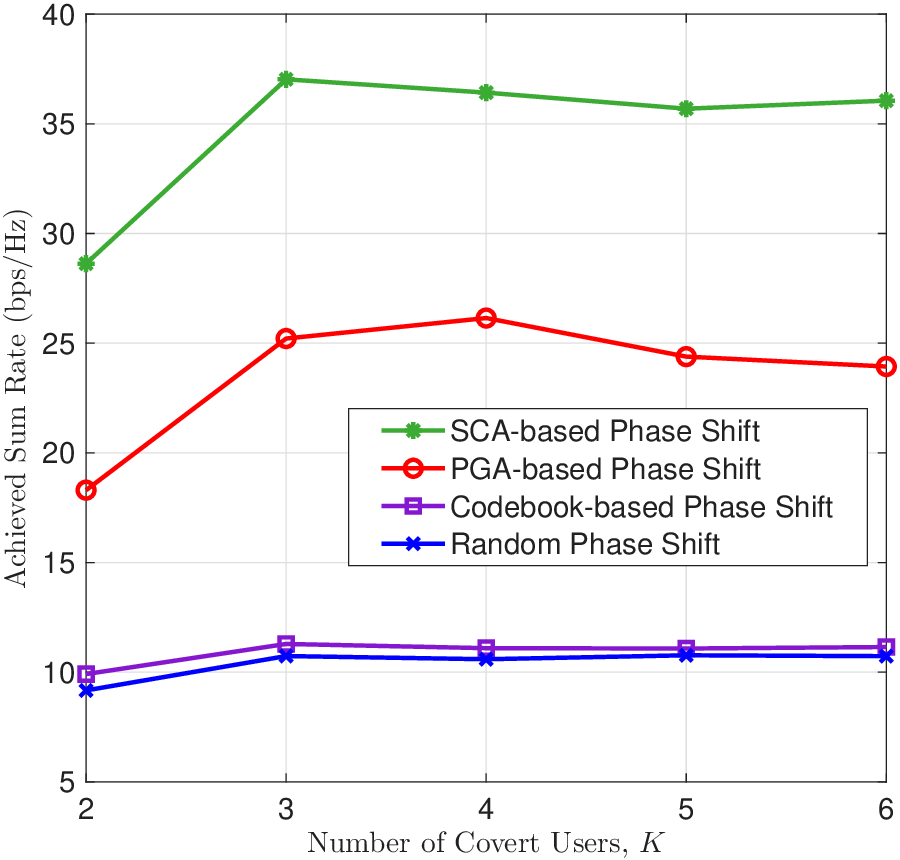}}
\caption{Sum rate performance of the proposed system against the number of covert users as compared with benchmarks.}
\label{res4}
\end{figure}

In Fig.~\ref{res4}, we evaluate the effect of increasing the number of covert UEs on the achieved sum rate. Since the number of UEs must be equal to or less than the number of transmitting antennas, we set the number of transmit antennas to $M=6$, and we change the number of UEs from $2$ to $6$. As we can see from the figure, the sum rate starts to increase with increasing the number of UEs up to a certain level, after which the sum rate starts to degrades with further increasing the number of UEs. This is because of the substantial increase in the inter-user interference as the number of UEs increases, where the beamfocusing design and the SIM phase shift design cannot effectively mitigate this interference. We can conclude that approximately for $K>4$, the system sum rate starts to gradually decrease for all phase shift design schemes. Moreover, the SCA-based phase shift design algorithm provides the highest performance among other algorithms, where it outperforms the sum rate achieved by PGA-based phase shift and codebook-based phase shift schemes with almost $47$\% and $228$\%, respectively.

The effect of increasing the number of wardens $U$ on the achieved sum rate for different phase shift design strategies is presented in Fig.~\ref{res5}. In general, the sum rate decreases as $U$ increases, revealing that adding more wardens introduces stricter constraints on the system that make it harder to optimize. The SCA-based phase shift consistently outperforms the other schemes, achieving the highest sum rate across all values of $U$, indicating its stronger ability to handle such strict covertness constraints. In comparison, PGA-based and codebook-based methods result in lower sum rates, while the random phase shift performs the worst. The performance gap becomes more evident as $U$ grows, showing that simpler or non-optimized designs struggle to cope with the increasing number of wardens. These results highlight the effectiveness and robustness of the SCA-based approach in dealing with systems with multiple wardens.
\begin{figure}[!t]
\centerline{\includegraphics[width=0.5\textwidth]{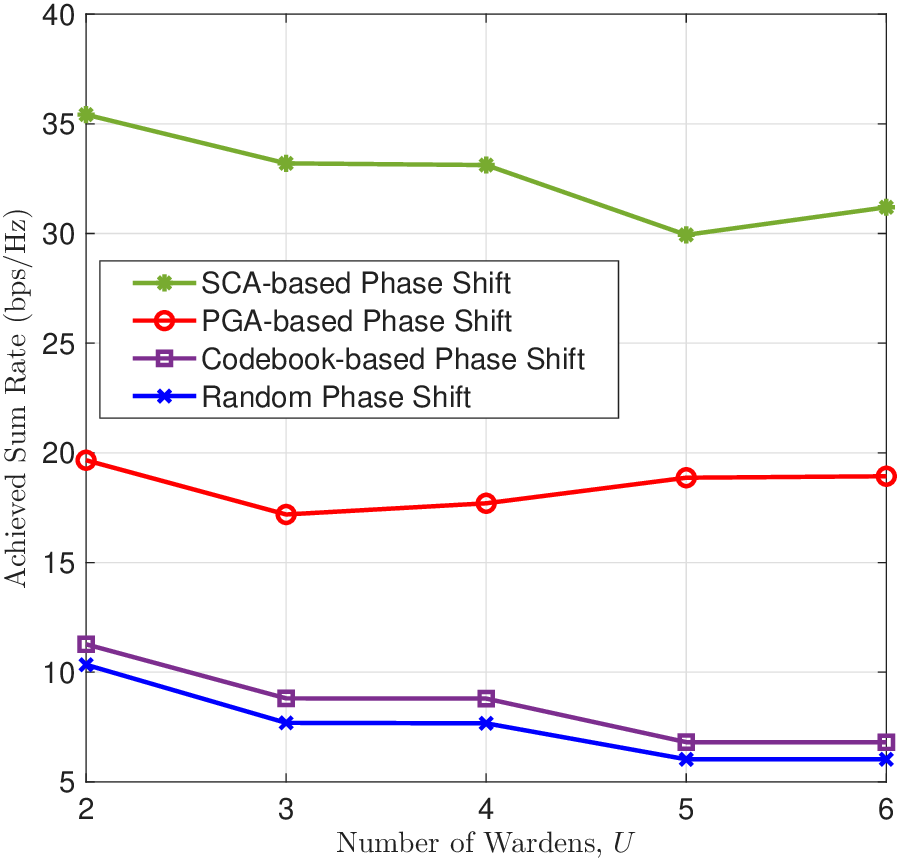}}
\caption{Sum rate performance of the proposed system against the number of wardens as compared with benchmarks.}
\label{res5}
\end{figure}

The effect of increasing the number of meta atoms on the achieved sum covert rate is studied in Fig.~\ref{res6}. As shown in the figure, the sum rate tends to increase as the number of meta atoms in each SIM layer increases for both the proposed algorithms as well as the benchmarks. This improvement can be attributed to the enhanced beamfocusing gain and higher array gain provided by a larger surface, which improves the received SINR at the legitimate users while maintaining the covert constraint. Moreover, the additional spatial degrees of freedom offered by more meta atoms enable more accurate phase shift design, which leads to more effective signal focusing and interference mitigation. However, the performance gain gradually diminishes and eventually saturates beyond a certain number of meta atoms. This behavior is because of the fact that, after a given surface size, the system becomes limited by other factors, such as transmit power. Consequently, increasing the number of meta atoms further provides a marginal beamforming improvement and yields only limited additional rate enhancement. Additionally, both the codebook-based and the random phase shift based algorithms offer almost the same performance, while the proposed SCA-based phase shift design algorithm outperforms other algorithms over the whole range of the number of meta atoms. For example, the SCA-based phase shift algorithm achieves an almost $103$\% higher sum rate compared to the PGA-based algorithm and an almost $270$\% higher sum rate compared to the codebook-based and random phase shift based algorithms at $45$ meta atoms. These percentages are further increased with increasing the number of meta atoms.
\begin{figure}[!t]
\centerline{\includegraphics[width=0.5\textwidth]{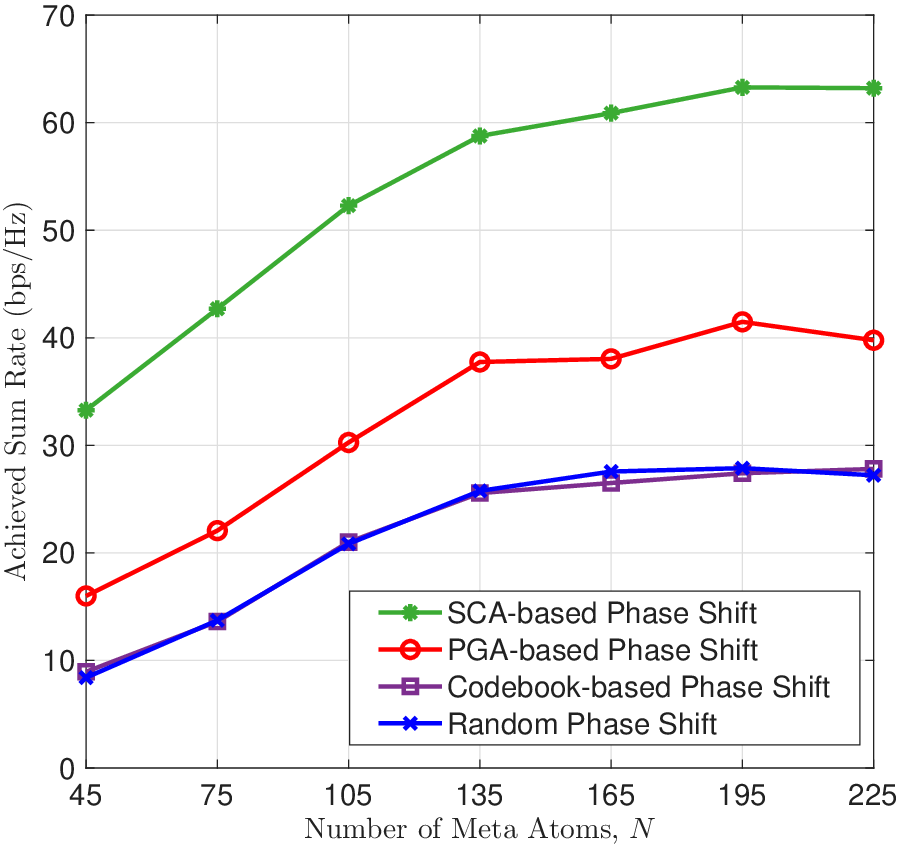}}
\caption{Sum rate performance of the proposed system against the number of meta atoms as compared with benchmarks.}
\label{res6}
\end{figure}

In Fig~\ref{res7}, the impact of increasing the covertness threshold (tolerance) $\epsilon$ and the number of wardens' observations $J$ on the performance of the proposed system is studied. It is clear from the figure that the achieved sum covert rate increases as the covertness threshold $\epsilon$ increases. This is because as $\epsilon$ increases, the covertness requirement of the system becomes easier, and therefore, the system tends to direct the resources towards increasing the SINR of the legitimate UEs. On the other hand, we can notice that there is a slight decrease or almost no decrease in the achieved sum covert rate as the number of wardens' observation increases, which can be attributed to the fact that increasing the number of observation increases the warden's capability of detection, and therefore, the system directs the resources to ensure the fulfillment of the covertness constraint, which reduces the achieved sum covert rate. However, the performance tends to saturate with further increasing the number of observations.
\begin{figure}[!t]
\centerline{\includegraphics[width=0.5\textwidth]{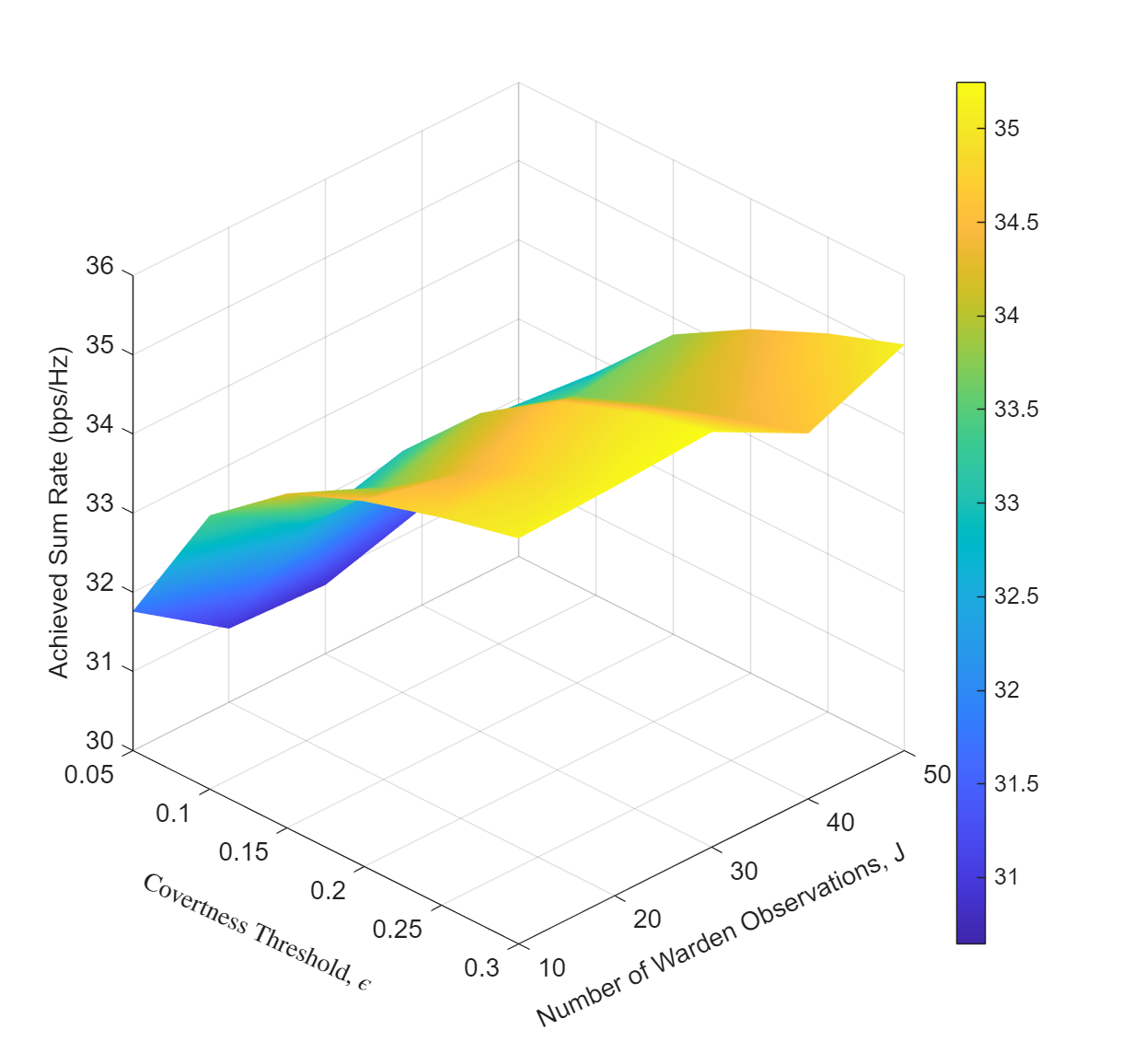}}
\caption{Sum rate performance of the proposed system as a function of the covertness threshold $\epsilon$ and the number of wardens' observations $J$.}
\label{res7}
\end{figure}
\section{Conclusion}
\label{sec_conc}
In this paper, we proposed a SIM-assisted multi-user MISO downlink covert communication system adopting near-field channel model. A sum covert rate maximization problem was formulated, where the digital beamfocusing vectors and the SIM phase shift matrices were optimized at all layers while fulfilling constraints of maximum power budget, minimum rate requirement, unit modulus phase shifts, and minimum detection probability at multiple wardens. Due to the highly non-convex nature of the formulated problem and the coupling between variables, AO-based algorithms were proposed, where the beamfocusing vectors and the SIM phase shift sub-problems are optimized alternately until convergence. We adopted SCA-based phase shift design as well as PGA-based phase shift design algorithms and compared the results with other benchmarks. The conducted simulations reveal that the SCA-based phase shift design method outperforms all other schemes with a considerable performance gap over different ranges of various system parameters, which indicates the robustness and effectiveness of the proposed method.

In future work, the proposed framework can be extended to scenarios investigating wideband and multi-cell near-field covert communication systems, which constitutes an interesting research direction. In addition, extending the proposed framework to secure integrated sensing and communications (ISAC) systems with joint sensing and covert communication is also of significant interest. Finally, developing low-complexity learning-based optimization frameworks for large-scale SIM architectures is another promising direction for real-time implementation.

{\appendix[Proof of Theorem~\ref{thm1}]
\label{App1}
Using the quotient rule, the partial derivatives of the SINR with respect to $\theta_n^l$ can be expressed as:
\begin{equation}
    \frac{\partial \gamma_k}{\partial \theta_n^l} = \frac{\partial}{\partial \theta_n^l} \frac{\text{num}_k}{\text{den}_k}= \frac{\text{den}_k \frac{\partial \text{num}_k}{\partial \theta_n^l}-\text{num}_k\frac{\partial \text{den}_k}{\partial \theta_n^l}}{\text{den}_k^2}, \label{Eq:App1_1}
\end{equation}
where $\text{num}_k =\left| \mathbf{h}_k^H \mathbf{G} \mathbf{v}_k\right|^2$ and $\text{den}_k =\sum_{i = 1, i\neq k}^K \left| \mathbf{h}_k^H \mathbf{G} \mathbf{v}_i\right|^2 + \sigma_n^2$. To find these derivatives for all $n\in\mathcal{N}$ and $l\in\mathcal{L}$, and according to~\eqref{Eq:1}, $\frac{\partial \mathbf{\Theta}^{l}}{\partial \theta_n^l} = \text{diag}\left\{0,\dots,je^{j\theta_n^l},\dots,0\right\} = je^{j\theta_n^l}\mathbf{E}_{n,n}$, where $\mathbf{E}_{n,n}$ is an $N\times N$ matrix whose entries are all zero except for the $(n,n)$-th element, which is equal to 1. Now, we get:
\begin{equation}
    \frac{\partial \mathbf{G}}{\partial \theta_n^l} = je^{j\theta_n^l} \mathbf{G}_L^l \mathbf{E}_{n,n} \mathbf{G}_R^l. \label{Eq:App1_2}
\end{equation}

Consequently, we can write:
\begin{dmath}
    \frac{\partial \text{num}_k}{\partial \theta_n^l} = 2\mathcal{R}\left[je^{j\theta_n^l} \left(\mathbf{h}_k^H\mathbf{G}\mathbf{v}_k\right)^*\mathbf{h}_k^H \mathbf{G}_L^l \mathbf{E}_{n,n} \mathbf{G}_R^l \mathbf{v}_k\right] = -2\mathcal{I}\left[e^{j\theta_n^l} \left(\mathbf{h}_k^H\mathbf{G}\mathbf{v}_k\right)^*\mathbf{h}_k^H \mathbf{G}_L^l \mathbf{E}_{n,n} \mathbf{G}_R^l \mathbf{v}_k\right]=-2\psi_{k,k}, \label{Eq:App1_3}
\end{dmath}
and
\begin{dmath}
    \frac{\partial \text{den}_k}{\partial \theta_n^l} = 2\sum_{i=1,i\neq k}^K\mathcal{R}\left[je^{j\theta_n^l} \left(\mathbf{h}_k^H\mathbf{G}\mathbf{v}_i\right)^*\mathbf{h}_k^H \mathbf{G}_L^l \mathbf{E}_{n,n} \mathbf{G}_R^l \mathbf{v}_i\right] = -2\sum_{i=1,i\neq k}^K \mathcal{I}\left[e^{j\theta_n^l} \left(\mathbf{h}_k^H\mathbf{G}\mathbf{v}_i\right)^*\mathbf{h}_k^H \mathbf{G}_L^l \mathbf{E}_{n,n} \mathbf{G}_R^l \mathbf{v}_i\right]=-2\sum_{i=1,i\neq k}^K\psi_{k,i}. \label{Eq:App1_4}
\end{dmath}

Substituting with~\eqref{Eq:App1_3} and~\eqref{Eq:App1_4} in~\eqref{Eq:App1_1}, we get~\eqref{Eq:28}, which completes the proof.
}


\balance

\bibliographystyle{IEEEtran}
\bibliography{IEEEabrv,References}

\newpage

 




\end{document}